\documentclass[lettersize,journal]{IEEEtran}
\usepackage{amsmath,amsfonts}
\usepackage{algorithm}
\usepackage{algpseudocode}
\usepackage{array}
\usepackage{amsmath}
\usepackage{subcaption}
\usepackage{booktabs}
\usepackage{textcomp}
\usepackage{stfloats}
\usepackage{url}
\usepackage{verbatim}
\usepackage{graphicx}
\usepackage{cite}
\usepackage{xcolor}
\usepackage{bm}
\usepackage{amsthm}
\usepackage{amssymb}
\usepackage{multirow}

\newtheorem{theorem}{Theorem}
\newtheorem{lemma}{Lemma}

\newtheorem{remark}{Remark}

\begin{document}

\title{Point-Cloud-Assistant Localized Statistical Channel Prediction by Tangent Gaussian Splatting}
\author{Ye~Xue,~\IEEEmembership{Member,~IEEE}, Yiheng~Wang,~\IEEEmembership{Graduate Student Member,~IEEE,} Xinhua Shao, Qi~Yan, Shutao~Zhang,~\IEEEmembership{Graduate Student Member,~IEEE}, and~Tsung-Hui~Chang,~\IEEEmembership{Fellow,~IEEE}
 \thanks{
    The authors have provided public access to their code at
https://github.com/chenjiadie/PC-TGS;\\

The work was supported in part by the National Key Research and
Development Program of China under Grant 2024YFA1014202; in part by the
National Natural Science Foundation of China under Grant 62301334, Grant
62031008; in part by the Shenzhen Science and Technology Program under Grant RCIC20210609104448114 and Grant ZDSYS20230626091302006; in part by Hetao Shenzhen-Hong
Kong Science and Technology Innovation Cooperation Zone Project (No.
HZOSWSKCCYB-2024016); in part by the Guangdong Major Project of
Basic and Applied Basic Research (No. 2023B0303000001); \\

   Ye Xue is with the School of Intelligent Systems Engineering, Shenzhen Campus of Sun Yat-sen University, Shenzhen 518107, China  (e-mail: xuey57@mail.sysu.edu.cn). \\
   Yiheng~Wang and Tsung-Hui~Chang  are with the Shenzhen Research Institute of Big Data,  and also with the School of Artificial Intelligence, The Chinese University of Hong Kong-Shenzhen, Guangdong 518172, China (email: yihengwang1@link.cuhk.edu.cn, tsunghui.chang@ieee.org).\\
   Xinhua Shao is with China Telecom, China. (shaoxh@chinatelecom.cn). \\
  Qi~Yan  and Shutao~Zhang  are with Networking and User Experience Lab, Huawei Technologies, China (email: yanqi1@huawei.com, zhangshutao2@huawei.com).\\
  }}
  
\maketitle

\begin{abstract}
Accurate, site-specific channel information is crucial for optimizing next-generation wireless networks. Among various approaches, localized statistical channel modeling (LSCM), which models the channel multipath angular power spectrum (APS) from the reference signal received power (RSRP) measurement, has emerged as a state-of-the-art method tailored for efficient network optimization. However, despite its effectiveness, LSCM  cannot predict APS at the vast majority of locations where no measurements are available, which significantly restricts its applicability in large-scale, real-world scenarios. To address this challenge, we present \emph{point-cloud-assisted tangent Gaussian splatting} (PC-TGS), the first framework to \emph{extrapolate} APS to unmeasured outdoor grids by integrating sparse radio measurements with dense LiDAR-based geometry. PC-TGS represents environmental scatterers as anisotropic 3D Gaussians, initialized and refined through a relaxed-mean reparameterization of the raw point cloud. A tangent-plane projection accurately maps each Gaussian into the local angular domain, while a depth-aware electromagnetic splatting process aggregates their contributions. To ensure practical deployment, we derive a closed-form Gaussian-weighted average (GWA) for APS bin integration and provide a provable error bound. { Evaluations on a LiDAR-scanned city-scale dataset (5M points, 6,310 RSRP samples) demonstrate that PC-TGS achieves better APS and RSRP prediction performance compared to state-of-the-art baselines  and faster inference time for APS extrapolation task. These results highlight the potential of PC-TGS to enable geometry-aware and data-efficient channel prediction in large-scale wireless digital twins.}
\end{abstract}

\begin{IEEEkeywords}
Channel Modeling, Radio Environment,  3D Gaussian Splatting, Point Cloud.
\end{IEEEkeywords}

\section{Introduction}

With the rapid development of next-generation mobile communication systems, wireless network optimization has become a fundamental approach for improving the quality of user experience~\cite{SRCON,li2022real}. Despite significant investments by network operators, network performance can deteriorate if system parameters are not properly tuned to the real environment. In modern networks, the challenge is compounded by the proliferation of parameters, including those related to antennas, handovers, beams, and carrier settings—all of which must be meticulously configured to fully realize the potential of wireless network systems.

Traditionally, network optimization has relied on engineering heuristics and repeated field measurements such as drive tests (DTs)~\cite{6353680}. However, these methods are time-consuming, costly, environmentally unsustainable, and limited to assessing only the current network configuration. To address these challenges and reduce operational risks, digital twin-based network optimization~\cite{10937314,10697404, 10.1145/3702644} has recently gained attention, enabling the modeling and simulation of wireless networks without manual intervention~\cite{jhzdt}. Despite these advances, a significant gap remains between the simulation and the real-world radio environment, primarily due to the difficulty of accurately capturing site-specific environmental effects and complex multi-path phenomena. Without precise modeling of the true radio environment, the performance gains achieved in simulation may fail to materialize in practice.

Existing channel models commonly used in 5G networks are not sufficient for this task. Geometry-based stochastic models (GBSM)~\cite{wang2021novel} are typically designed for broad scenario types, such as indoor, outdoor, rural, or urban, but lack the ability to represent localized environmental features. Deterministic models, such as ray tracing~\cite{10.1109/MCOM.001.2200486}, provide site-specific accuracy but require detailed map information and incur prohibitive computational costs. Simplified empirical models, such as the COST231-Hata model~\cite{1543252}, are computationally efficient but cannot characterize multi-path effects or sufficient channel properties for network optimization.

To address these limitations, \emph{localized statistical channel modeling} (LSCM)~\cite{LSCM,GNNLSCM} has emerged as a promising alternative tailored for efficient network optimization. Unlike traditional approaches that rely on high-dimensional channel impulse response (CIR) measurements, LSCM uses low-overhead reference signal received power (RSRP) data from multiple transmit beams obtained from the network layer. This enables the efficient estimation of the angular power spectrum (APS) at measured locations, revealing localized multi-path structure and environmental topography, while significantly reducing data collection requirements.

Despite its effectiveness, a critical and often misunderstood limitation of LSCM is its lack of spatial generalization: it can only reconstruct APS at locations where RSRP measurements are directly available. LSCM cannot interpolate or extrapolate APS to unmeasured regions. While one can interpolate RSRP at unmeasured points using spatial methods (e.g., Kriging~\cite{sato2020radio}) and then apply LSCM, this approach enables only mild interpolation and relies on the presence of nearby measurements. True extrapolation that predicts APS in regions with no measurements at all remains fundamentally unsolved and requires strong environmental priors and model generalization. This limitation significantly restricts the utility of LSCM for large-scale simulation, network planning, and digital twin construction.

To overcome this barrier, we introduce \emph{localized statistical channel prediction} (LSCP): extrapolating the APS at unmeasured locations by leveraging both low-overhead RSRP data and environmental priors (such as  LiDAR point clouds). Unlike LSCM, LSCP aims to infer the spatial distribution of APS everywhere, including regions without any signal measurements. In LSCM, APS is recovered at a location with its own RSRP, but in LSCP, the model must predict APS at completely unmeasured locations, transferring knowledge from measured regions to the unknown. This demands a model capable of true spatial extrapolation, robust generalization, and effective use of environmental structure.

{\subsection{Related Work}

To contextualize the LSCP problem, we review three lines of related research. We begin with LSCM and sparse recovery methods that estimate APS at measured locations, then discuss radio map prediction methods that attempt spatial generalization of signal strength, and finally examine physics-based and neural scene representations that model wireless propagation.

\subsubsection{APS Recovery at Measured Locations}

LSCM~\cite{LSCM,GNNLSCM} recovers the APS $\bm{x}$ from low-overhead RSRP measurements by solving an ill-conditioned linear inverse problem $\bm{y}=\bm{A}\bm{x}$ at locations where RSRP $\bm{y}$ and system parameters $\bm{A}$ are known. Within this framework, several sparse recovery algorithms can be adopted. WNOMP~\cite{LSCM} uses greedy orthogonal matching pursuit, while  {temporally correlated MSBL (TMSBL)~\cite{zhang2011sparse}} provides a principled probabilistic framework for the multiple measurement vector (MMV) problem, offering improved accuracy by grouping spatially or signal-similar samples to exploit joint sparsity with a higher computational cost. Although effective at measured locations, all these methods fundamentally operate only where measurements exist. Hence, they cannot generalize APS predictions to unmeasured regions. Even when combined with spatial interpolation of RSRP (e.g., via Kriging), the recovered APS remains constrained to mild interpolation and degrades rapidly in measurement-free areas.

\subsubsection{Spatial Radio Map Prediction}

A complementary line of research focuses on predicting radio signal strength at unmeasured locations. Classical spatial interpolation methods such as inverse distance weighting (IDW)~\cite{kuo2010discriminant} and Kriging~\cite{sato2020radio} estimate RSRP using proximity-weighted averaging, while deep learning approaches such as RadioUNet~\cite{levie2021radiounet} leverage convolutional neural networks to predict path-loss maps from  building geometry. However, all these methods predict signal power values rather than the high-dimensional channel APS, and thus cannot capture the angular multipath structure essential for network optimization tasks such as beamforming and antenna configuration. Consequently, they address a  less informative prediction target than the proposed LSCP.

\subsubsection{Physics-Based and Neural Scene Representation Methods}

Deterministic ray tracing~\cite{eertmans2025differt} provides site-specific APS prediction by simulating electromagnetic wave propagation through a geometric scene model. While physics-grounded and capable of producing the full APS, ray tracing relies on precise material parameters and incurs high per-query computational cost, limiting its scalability for large-scale network optimization.

Recent advances in neural scene representation have adapted techniques from computer vision to wireless channel modeling. NeRF-based models such as NeRF$^2$~\cite{nerf2} and NeWRF~\cite{newrf} reconstruct radio maps or channel impulse responses (CIRs) by capturing fine spatial and multipath effects. However, these approaches are computationally intensive and, crucially, require dense ground-truth APS or CIR labels for training which can  rarely be satisfied in practical outdoor deployments. To improve efficiency, 3D Gaussian Splatting (3DGS)~\cite{3dgs} and its wireless adaptation WRF-GS~\cite{wen2025wrf} offer real-time performance with anisotropic Gaussian basis functions~\cite{3dgs2,3dgs3,3dgs4}, but still assume access to dense APS data for training and do not leverage environmental geometry for beamspace representation, limiting their applicability to controlled indoor settings.

In summary, no existing method can directly address the LSCP task which requires extrapolating APS at unmeasured locations using only low-dimensional sparse RSRP and environmental priors. }

\subsection{Contribution}

To address the aforementioned challenges, we propose a unified, environment-aided 3DGS-based framework called \emph{point-cloud-assisted tangent Gaussian splatting} (PC-TGS), which enables geometry-aware APS extrapolation at unmeasured outdoor locations by integrating sparse radio measurements with dense LiDAR-derived environmental priors. The main contributions of this work are as follows:

\begin{itemize}
    \item \textbf{localized statistical channel prediction (LSCP) framework based on 3D Gaussian Splatting (3DGS):} We formally define the localized statistical channel prediction (LSCP)
    problem and propose PC-TGS, a framework that repurposes 3D Gaussian
    splatting for the radio domain with a rigorous physical foundation.
    Each 3D Gaussian ellipsoid serves as a \emph{physical scatterer proxy},
    whose position and anisotropic covariance $(\boldsymbol{\mu},
    \boldsymbol{\Sigma})$ encode both the location and effective scattering
    area of an environmental object, while its spherical harmonic (SH) coefficients capture
    angle-dependent scattering intensity consistent with electromagnetic
    scattering principles. Furthermore, the splatting accumulation is
    reinterpreted as \emph{multipath superposition}: Gaussian opacities
    represent per-scatterer signal attenuation, and the rendering
    aggregation sums multipath components arriving at the receiver,
    replacing the RGB alpha blending of visual 3DGS with a physically
    meaningful signal composition pipeline. This dual reinterpretation
    establishes that PC-TGS is not a direct application of vision-based
    3DGS but a fundamentally redesigned framework grounded in radio
    propagation physics.
    \item \textbf{Efficient point cloud-based scatterer modeling:} We tackle the significant challenges of using real-world LiDAR data, which is typically both extremely dense and inherently noisy. Our proposed relaxed-mean reparameterization algorithm efficiently selects and refines a compact set of representative virtual scatterers directly from the raw point cloud, filtering out irrelevant and noisy points. This process preserves the most informative geometric features for radio propagation modeling, suppresses measurement noise, and reduces the computational load, ultimately enabling PC-TGS to robustly bridge real environmental structure with wireless multipath effects even in the presence of imperfect sensing data.

    \item \textbf{Tangent-plane projection and GWA-based APS synthesis:} We introduce a scalable tangent-plane projection that rigorously aligns environmental scatterers with the antenna angular domain, which is fundamentally different from the camera-based projections in 3DGS used in vision tasks. To further ensure efficient and accurate APS computation, we derive a closed-form Gaussian weighted average (GWA) bin integration formula, together with a provable error bound, enabling fast and mathematically principled APS synthesis on large-scale datasets.

    \item \textbf{Unified, robust, and scalable pipeline:} By jointly integrating sparse radio measurements, dense point cloud priors, and differentiable electromagnetic modeling into a single end-to-end framework, PC-TGS achieves robust and physically interpretable APS prediction in realistic, data-sparse environments. Extensive experiments on city-scale datasets demonstrate that PC-TGS consistently outperforms state-of-the-art baselines in both accuracy and robustness, while delivering practical inference speed.
\end{itemize}

The remainder of this paper is organized as follows.  
Section~\ref{sec:system} summarizes the preliminaries of LSCP and reviews the 3DGS technique.  
Section~\ref{sec:propose} details the proposed PC-TGS framework, including relaxed-mean reparameterization, attribute construction, tangent-plane projection, electromagnetic splatting,  {and derives an efficient GWA-based integration formula with a provable error bound}.  
Section~\ref{sec:experiments} reports extensive experiments that compare PC-TGS with state-of-the-art alternatives.  
Finally, Section~\ref{sec:conclusion} concludes the paper and outlines promising future research directions.

\textit{Notations:} {$\mathbb R$, $\mathbb C$ and $\mathbb{B}$ stand for the set of real numbers, the set of complex numbers, and the set of binary numbers (0 and 1), respectively. $[N]$ is the set of integers $\{1,2,\ldots,N\}$. The operators $(\cdot)^\textnormal{T}$, $\|\cdot\|_0$, $\|\cdot\|_1$, and $\|\cdot\|_2$ represent transpose, $\ell_0$-Norm, $\ell_1$-Norm, and the $\ell_2$-Norm respectively. $\mathbb E[\cdot]$ denotes the expectation.  The \emph{multivariate} Gaussian distribution with mean vector $\bm{\mu} \in \mathbb{R}^d$ and covariance matrix $\bm{\Sigma} \in \mathbb{R}^{d \times d}$ is denoted by $\mathcal{N}(\bm{x};\bm{\mu},\bm{\Sigma})$.}

\section{Preliminaries}\label{sec:system}

\subsection{Localized Statistical Channel Prediction}
\label{Subsection: LSCM}

Accurately predicting the spatial distribution of wireless channel characteristics is a fundamental challenge in next-generation network design. 
Fig.~\ref{fig:Task} illustrates the overall system scenario and the channel prediction problem considered in this work.

\begin{figure}
\begin{center}
\centerline{\includegraphics[width=1\columnwidth]{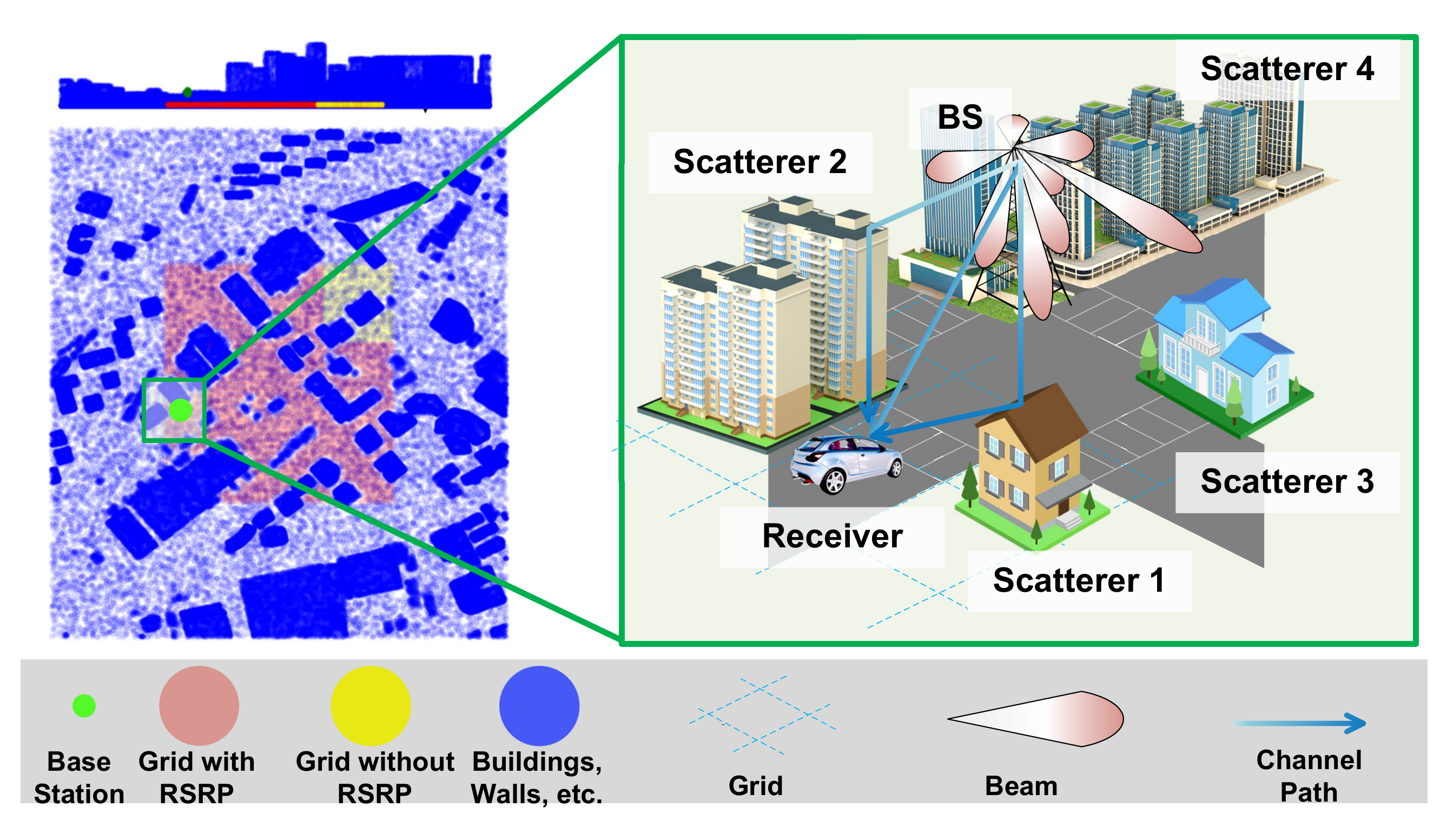}}
\caption{%
Illustration of the system scenario and channel modeling task.
\textbf{Left:} Front and top views of a 3D point cloud of an outdoor environment in City A of China, showing the spatial distribution of buildings and measurement grids. Grids with available RSRP measurements (red) and unmeasured grids (yellow) are highlighted, along with the location of the base station (green).
\textbf{Right:} The base station broadcasts the multi-beam reference signals, which are then distorted by major environmental scatterers. Then, the receiver obtains the reference signal and reports the corresponding signal power. 
}
\label{fig:Task}
\end{center}
\end{figure}

Specifically, we consider a wireless communication network where the base station (BS) is equipped with a uniform planar array (UPA) of $N_T = N_X \times N_Y$ antennas. The tilt and azimuth angles of departure (AoDs) are discretized uniformly into $N_V$ and $N_H$ values, denoted as $\theta_p$ ($1 \le p \le N_V$) and $\varphi_q$ ($1 \le q\le N_H$), respectively. The downlink channel coefficient for the $(x, y)$-th antenna is modeled as follows~\cite{38901}:
\begin{align}
	h_{x, y} = & \sum_{p = 1}^{N_V} \sum_{q = 1}^{N_H}  \sqrt{\beta_{p,q} } \times g_{p, q} \times e^{-j2\pi\frac{d_x x}{\lambda}\cos \theta_p\sin \varphi_q}  \nonumber\\
    &  \quad \qquad \times e^{-j2\pi\frac{d_y y }{\lambda}\sin \theta_p} \times e^{-j\omega_{p,q}-j\omega_{x,y}},
\end{align}
where $\beta_{p,q}$ is the channel power of the path in the AoDs $(\theta_p,\varphi_q)$, $g_{p,q}$ is the antenna gain, $\lambda$ is the carrier wavelength, $d_x$ and $d_y$ denote antenna spacings, and $\omega_{p,q}$, $\omega_{x,y}$ represent random phase errors.

In downlink, the BS transmits $M$ directional reference signal beams that sweep across the angular domain. The precoding matrix for the $m$-th beam is $\boldsymbol{B}^{(m)} \in \mathbb{C}^{N_X \times N_Y}$, with $B_{x, y}^{(m)} = e^{j\phi_{x,y}^{(m)}}$. The received signal strength (RSRP) for beam $m$ at a user location is then
\begin{align}
	rsrp^{(m)} =  P_{tx}\left|\sum_{x=1}^{N_X}\sum_{y=1}^{N_Y} h_{x, y} B_{x, y}^{(m)}\right|^{2},
\end{align}
where $P_{tx}$ is the transmit power.

Note that $rsrp^{(m)}$ depends on three random variables: $\omega_{p,q}$, $\omega_{x,y}$, and $\beta_{p,q}$. By taking the expectation over the random phase errors, we relate the measured RSRP to the angular power spectrum (APS) $\bm{x} = \mathbb{E}[\boldsymbol{\beta}]$, where $\boldsymbol{\beta} = [\beta_{1,1},\dots,\beta_{N_V,N_H}]^T \in \mathbb{R}^{N_A}$ and $N_A = N_V \times N_H$. Letting $\bm{rsrp} = [rsrp^{(1)},\dots,rsrp^{(M)}]^T \in \mathbb{R}^M$, the expected multi-beam RSRP $\bm{y}$ and APS $\bm{x}$ satisfy~\cite{LSCM} the relationship
\begin{equation}
    \begin{aligned}\label{eq:y=Ax}
	{\bf y} = \mathbb{E}[\bm{rsrp}] = \bm{A} {\bm x},
\end{aligned}
\end{equation}
where $\bm{A} \in \mathbb{R}^{M \times N_A}$ is the measurement matrix determined by the antenna and beam settings. In practice, the sample mean is used to approximate $\bm{y}$.

Given the sparse nature of wireless channels due to limited scattering~\cite{10521790}, classical LSCM seeks, at measured location grids $l \in \mathcal{L}_k$, to estimate the APS $\bm{x}_l$ by solving
\begin{equation}
\begin{aligned}
     \min_{\bm{x}_l}\quad &  \|\bm{A} \bm{x}_l - \bm{y}_l\|_2^2 \\
    \text{s.t.}\quad & \|\bm{x}_l\|_0 \leq P,\\
    &  \bm{x}_l \geq 0,\qquad l \in \mathcal{L}_k,
\end{aligned}
\end{equation}
where $\bm{y}_l$ is the measured RSRP and $P$ is the maximum number of channel paths.

A key advantage of LSCM is its ability to separate the site-specific channel statistics (APS $\bm{x}_l$) from the system configuration (antenna/beam settings, encoded in $\bm{A}$). Since the APS $\bm{x}_l$ is determined solely by the physical environment, it remains invariant under different BS configurations, provided the environment does not change. This property enables a powerful application that, once the APS $\bm{x}_l$ is estimated from RSRP measurements under one set of system parameters $\bm{A}$, one can predict the RSRP under a different set of parameters $\bm{A}^r$ by simply computing
\begin{equation}
    \bm{y}_l^r = \bm{A}^r \bm{x}_l.
\end{equation}
This capability is fundamental to network optimization since it allows for the simulation and evaluation of new antenna or beam configurations without requiring new drive tests or channel measurements, provided that a reliable APS estimate is available. 

In practice, RSRP can only be measured at a sparse subset of spatial grids $\mathcal{L}_k$ (e.g., via drive tests), leaving the APS unknown at most locations $\mathcal{L}_o$ (such as urban blocks or private areas), as illustrated in Fig.~\ref{fig:Task}. This motivates the problem of \emph{localized statistical channel prediction} (LSCP) to predict the APS $\hat{\bm{x}}_l$ at all unmeasured grids $l \in \mathcal{L}_o$, given
\begin{equation}
\begin{aligned}
       &\text{Inputs:}~\{\bm{y}_l\}_{l\in \mathcal{L}_k},\ \bm{A},\ \text{priors}\\
   & \text{Output:}~\{\hat{\bm{x}}_l\}_{l\in \mathcal{L}_o}.
\end{aligned}
\end{equation}
In other words, LSCP seeks a mapping
\begin{equation}\label{eq:taskover}
    F: (\{\bm{y}_l\}_{l\in \mathcal{L}_k},~\bm{A},~\text{priors}) \rightarrow \{\hat{\bm{x}}_l\}_{l\in \mathcal{L}_o}
\end{equation}
that infers the spatial distribution of APS at all unmeasured locations using only limited RSRP and environmental knowledge.

 {\begin{remark}[Practical Scenario Justification]
The assumed scenario where dense point clouds paired with sparse, low-dimensional RSRP measurements is practically well motivated for the following reasons.
\begin{enumerate}
    \item  In network layer,  only RSRP is accessible via the Minimization of Drive Tests (MDT) or drive test (DT) protocol ~\cite{hapsari2012mdt}.  Each RSRP measurement yields a vector of dimension $M$ (the number of reference signal beams, typically $1\sim 8$), whereas the APS lives in $\mathbb{R}^{N_A}$ with thousands potential angular bins.  Because $M \ll N_A$, recovering the high-dimensional APS from the low-dimensional RSRP is an inherently ill-posed inverse problem that must exploit the sparsity of the wireless channel.
    \item  RSRP is an in-situ measurement which means  the receiver can only report signal strength at its exact physical location.  Drive-test vehicles are confined to road networks, therefore, RSRP samples are available only at scattered points along accessible routes, leaving the vast majority of the service area unmeasured.
   \item In contrast, LiDAR is a remote sensing modality that captures 3D geometry over a large volume per scan (radius up to 150\,m per pass~\cite{Sun2019ScalabilityIP}).  A single data-collection vehicle therefore produces dense environmental geometry over a wide area while simultaneously collecting only trajectory-constrained radio measurements.  Even when dedicated LiDAR is unavailable, image-based reconstruction methods such as DUSt3R~\cite{wang2024dust3r} can produce point clouds of sufficient accuracy (0.1--0.5\,m) for our framework.
\end{enumerate}
This asymmetry between dense geometry and sparse, low-dimensional radio data is therefore natural and consistent with modern network optimization workflows.
\end{remark}}

\subsection{3D Gaussian Splatting (3DGS)\label{sec:3dgsback}} 3DGS \cite{3dgs} is an explicit radiance field-based scene representation in computer vision that models a radiance field using a large number of total $I$ 3D anisotropic balls, { the $i$-th 3D Gaussian is defined as:
\begin{align}
    G_i(\bm{z}) = e^{-\frac{1}{2}(\bm{z} - \bm{\mu}_i)^\textnormal{T} {\bf \Sigma}^{-1}_i (\bm{z} - \bm{\mu}_i)},
\end{align}}where $\bm{z} = (z_0, z_1, z_2)$ denotes the three-dimensional spatial coordinates of the Gaussian point and $\bm{\mu}_i \in \mathbb{R}^3$ and ${\bf \Sigma}_i\in \mathbb{R}^{3\times3}$ denote the mean and variance of the Gaussian representation.
Besides, the $i$-th anisotropic ball is also characterized by the Spherical Harmonics (SH) coefficients $\bm{\tau}_i \in \mathbb{R}^{3 \times (k+1)^2}$ (where $k$ is the degree of the SH function $\bm{sh}(\cdot, k)$) \cite{fridovich2022plenoxels} for modeling view-dependent color. { Specifically, the color of the $i$-th 3D Gaussian at the view direction $\bm v_i$ is
\begin{align}
    \bm{c}_i= \bm{sh}(\bm{\tau}_i,\bm v_{i}, k) \in \mathbb{R}^3.
\end{align}{During the rendering phase, the 3D Gaussians $\{G_i(\bm{z})\}_{i=1}^{I}$
are projected (rasterized) onto the image plane, forming 2D Gaussian splats, which can be expressed as  $\{\bm {w}_i(\kappa)=\text{Proj}(G_i(\bm{z}))\}_{i=1}^I$ at the $\kappa$-th pixel, as described in \cite{3dgs}. The 2D Gaussian splats are sorted from front to back tile-wisely, and the \textit{alpha}-blending  algorithm is applied to approximate the $\kappa$-th pixel RGB value $\bm{C}(\kappa) \in \mathbb{R}^3$ as 
\begin{align}
    \bm{C}(\kappa) = \sum_{i\in\mathbb I_\kappa} \bm{c}_i \alpha_i\cdot\bm {w}_i(\kappa)\prod_{j=1}^{i-1} (1 - \alpha_j\cdot\bm {w}_j(\kappa)),
    \label{eq:alpha}
\end{align}
where $\mathbb I_{\kappa}$ is the set of splats that contribute to the $\kappa$-th pixel, {   $\alpha_i \in \mathbb{R}$ is the opacity of the $i$-th anisotropic 3D Gaussian ball.} After computing the color for each pixel, the rendered image is compared to the ground-truth image to calculate the pixel-wise loss. This loss is then used to update the model's properties parameters $\bm{\Phi} = \{(\bm{\mu}_i, \bm{\Sigma}_i,\alpha_i, \bm {\tau}_i ), i \in [I]\}$, ensuring the rendered scene aligns closely with the target, where $[I]$ denotes the vector consisting of $[1, \cdots, I]$.

\section{Proposed Framework: Point-Cloud-Assisted Tangent Gaussian Splatting}\label{sec:propose}
\subsection{Framework Overview}

\begin{figure*}[t]
\centering
    \includegraphics[width=\textwidth]{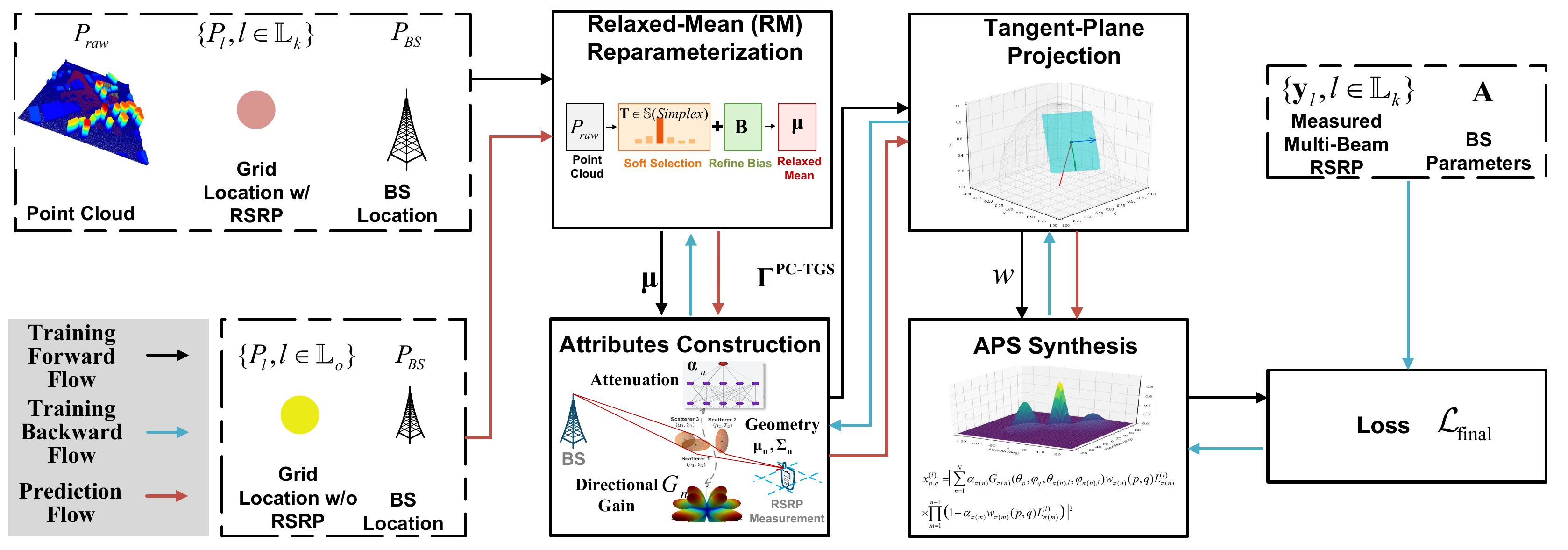}
    \caption{
    \textbf{Overview of the proposed PC-TGS (Point-Cloud-Assisted Tangent Gaussian Splatting) framework.}
    Given sparse, low-dimensional RSRP measurements, base station information, and dense 3D point clouds, PC-TGS extracts a set of mean positions of the virtual scatterers via relaxed-mean reparameterization, constructs physical and signal attributes, projects them onto the local tangent plane at each angular bin, and synthesizes the APS via differentiable electromagnetic splatting. This pipeline enables accurate, environment-aware channel prediction at unmeasured locations.
    }\label{fig:Framework}
\end{figure*}

Building upon the LSCP problem setting in Section~\ref{sec:system}, our goal is to enable accurate APS prediction at unmeasured locations by leveraging both sparse RSRP measurements and detailed environmental priors. To address the spatial extrapolation challenge, we develop a unified, physically-consistent PC-TGS framework.

PC-TGS extends 3DGS to the wireless domain by designing virtual scatterers using 3D Gaussians to represent the influence of the environment on the radio propagation. PC-TGS integrates (1) sparse RSRP measurement {  $\bm{y}_l\in \mathbb{R}^M$} with location {  $\bm{P}_l\in \mathbb{R}^3$} from drive tests or reports $\{\bm{y}_l,\bm{P}_l\}_{l\in \mathcal{L}_k}$, (2) known base station locations $\bm{P}_{BS}$ and BS parameters $\bm{A}$, and (3) dense 3D point cloud data $\bm{P}_{\mathrm{raw}} \in \mathbb{R}^{R \times 3}$ describing environmental geometry to provide a physical interpretable modeling to address LSCP.

As illustrated in Fig.~\ref{fig:Framework}, the PC-TGS pipeline consists of four core modules:
\begin{enumerate}
    \item \textbf{Relaxed-Mean Reparameterization:} A relaxed-mean reparameterization algorithm selects the position means of $N$ 3D Gaussians to represent the location of key virtual scatterers from the dense point cloud, reducing complexity while preserving key geometric features.
    \item \textbf{Attributes Construction:} Each virtual scatterer is assigned a set of learnable physical and electromagnetic attributes (e.g., location, shape, attenuation, and complex gain), modeling its interaction with the wireless channel and enabling APS synthesis.
    \item \textbf{Tangent-Plane Projection:} For each angular bin, all virtual scatterers are projected onto the local tangent plane that aligns environmental geometry with the angular domain used for APS computation.
    \item \textbf{APS Synthesis:} The projected scatterers' contributions are aggregated via a depth-aware, differentiable splatting rule, producing the predicted APS at each location grid in a way that is both physically grounded and amenable to end-to-end learning.
\end{enumerate}

To realize the LSCP objective defined in~(\ref{eq:taskover}), PC-TGS operates via two key mappings:
\begin{subequations}
\begin{align}
     &F^{\text{train}}: {(\{\bm{y}_l, \bm{P}_l\}_{l \in \mathcal{L}_k},\, \bm{A},\, \bm{P}_{\mathrm{raw}},\, \bm{P}_{BS})}
     \mapsto \bm{\Gamma}^{\text{PC-TGS}}, \\
   & F^{\text{predict}}_{\bm{\Gamma}^{\text{PC-TGS}}}:{(\bm{P}_{BS},\, \{\bm{P}_l\}_{l \in \mathcal{L}_o})}
   \mapsto \{\hat{\bm{x}}_l\}_{l \in \mathcal{L}_o},
\end{align}
\end{subequations}
where $\bm{\Gamma}^{\text{PC-TGS}}$ is the set of all learnable model parameters, and $\hat{\bm{x}}_l\in \mathbb{R}^{N_A\times 1}$ is the predicted APS at location $l$ (with $N_A$ angular bins). By tightly coupling sparse measurements with LiDAR-based geometry and a physics-inspired differentiable pipeline, PC-TGS enables accurate, scalable APS extrapolation to address the core challenge of LSCP.

 \subsection{Relaxed-Mean Reparameterization for Virtual Scatterer Extraction}\label{sec:RM}

A crucial step in the PC-TGS framework is the extraction of the position of 3D Gaussians that represent key virtual scatterers, such as building outlines or major reflectors, that dominate multipath propagation. Given a dense and noisy 3D point cloud $\bm{P}_{\mathrm{raw}}\in \mathbb{R}^{R\times 3}$, where $R$ is the number of raw points, the objective is to select $N$ representative scatterer locations $\{\bm{\mu}_n\}_{n=1}^N$ that are both physically plausible and maximally informative for APS modeling.

Ideally, this selection can be formulated as a constrained optimization problem:
\begin{subequations}
\begin{align}
&\text{Find}~\{\bm{\mu}_n\}_{n=1}^N \in \mathbb{R}^{N \times 3}, \\
&\text{s.t. }\quad  \bm{P}_N \subset \bm{P}_{\mathrm{raw}},  \\
&\qquad \|\{\bm{\mu}_n\}_{n=1}^N - \bm{P}_N\|_2 \leq \epsilon_1,  \\
&\qquad \|\bm{P}_N - \bm{P}_{\mathrm{BS}}\|_2 \leq \epsilon_2,
\end{align}
\end{subequations}
where $\bm{P}_N$ is a subset of the raw point cloud, the second constraint enforces proximity of the final scatterer positions to their selected anchors, aiming to mitigate noise, and the third  constraint prioritizes scatterers closer to the BS, as they contribute
to dominant signal paths.

However, this combinatorial selection is non-differentiable and thus not suitable for gradient-based learning. To address this, PC-TGS adopts a relaxed-mean (RM) reparameterization, in which each scatterer position is expressed as a soft, trainable aggregation of the raw point cloud:
\begin{equation}
    \{\bm{\mu}_n\}_{n=1}^N = \bm{T} \bm{P}_{\mathrm{raw}} + \bm{B},
\end{equation}
where $\bm{T} \in \mathbb{R}^{N \times R}$ is a trainable assignment matrix and $\bm{B} \in \mathbb{R}^{N \times 3}$ is a bias term for denoising and refinement. Each row $\bm{T}_{n,:}$ assigns non-negative weights to the $R$ points, subject to the simplex constraint,
\begin{equation}
\bm{T} \in \mathbb{S} = \left\{\bm{T} \in \mathbb{R}^{N \times R}:~ T_{n,r} \geq 0,~ \sum_{r=1}^R T_{n,r} = 1,~\forall n\right\}.
\end{equation}
This formulation allows $\bm{T}\bm{P}_{\mathrm{raw}}=\bm{P}_N$ to act as a soft selection of anchor points for each scatterer. To ensure that each scatterer is associated predominantly with a single region or point, a maximum element constraint (MEC) is imposed, requiring the largest entry in each row of $\bm{T}$ to strongly dominate all others:
\begin{equation}
    \max_{r} T_{n,r} \gg T_{n,r'},~ \forall n,~ r' \neq \arg\max_{r} T_{n,r}.
\end{equation} With this reparameterization, the original hard constraints are relaxed and encoded as regularization terms in the objective. Specifically, the proximity constraint $\|\bm{P}_N - \bm{P}_{\mathrm{BS}}\|_2 \leq \epsilon_2$ is promoted by the penalty $\lambda_1 \|\bm{T} \bm{P}_{\mathrm{raw}} - \bm{P}_{\mathrm{BS}}\|_2^2$, which favors scatterers near the BS. The denoising constraint $\|\{\bm{\mu}_n\}_{n=1}^N - \bm{P}_N\|_2 \leq \epsilon_1$ is enforced by the term $\lambda_2 \|\bm{B}\|_2^2$, ensuring the bias remains small and $\bm{\mu}_n$ stays close to its anchor. The MEC requirements are promoted by $\lambda_{\mathrm{MEC}}\sqrt{\sum_{n=1}^N\max_{r\in[R]}(T_{n,r}-1)^2}+\lambda_{\mathrm{sparsity}}\|\bm{T}\|_1$, which penalizes deviation from a one-hot vector in each row.

Then, the resulting optimization problem for the RM module is given by
\begin{align}
\min_{\bm{T}, \bm{B},\, \Phi^{\mathrm{PC-TGS}} \setminus \{\bm{\mu}_n\}}~ & \mathcal{L}_{\text{APS}} + \mathcal{L}_{\text{reg}}, \\
\text{s.t.}~ & \{\bm{\mu}_n\}_{n=1}^N = \bm{T} \bm{P}_{\mathrm{raw}} + \bm{B},~ \bm{T} \in \mathbb{S},
\end{align}
where the regularization term is
\begin{align}\label{eq:RMt}
\mathcal{L}_{\text{reg}} = &~ \lambda_1 \|\bm{T} \bm{P}_{\mathrm{raw}} - \bm{P}_{\mathrm{BS}}\|_2^2 + \lambda_2 \|\bm{B}\|_2^2 \nonumber \\
&~ + \lambda_{\mathrm{MEC}}\sqrt{\sum_{n=1}^N\max_{r\in[R]}(T_{n,r}-1)^2}+ \lambda_{\mathrm{sparsity}}\|\bm{T}\|_1,
\end{align}
with each term corresponding directly to the respective relaxed constraint above. $\mathcal{L}_{\text{APS}}$ will be specified in (\ref{eq:finalloss}).

This relaxed yet physically motivated approach to virtual scatterer selection provides a robust geometric foundation for subsequent APS attribute learning and tangent-plane projection, thus enabling PC-TGS to bridge raw sensing data and statistical channel prediction in complex environments.

\subsection{Attributes Construction}

Once the positions of the $N$ key virtual scatterers $\{\bm{\mu}_n\}_{n=1}^N$ have been extracted, PC-TGS models the physical interaction between each scatterer and the wireless propagation channel using a set of learnable signal attributes. These attributes are designed to capture both the geometric and electromagnetic effects that influence the local angular power spectrum (APS) at each location.

\textbf{PC-TGS Attributes:}  
Considering the channel reciprocity, each virtual scatterer $n$ is parameterized by
\begin{itemize}
    \item $\bm{\mu}_n \in \mathbb{R}^3$, which is the spatial position of the $n$-th scatterer, extracted via RM in subsection \ref{sec:RM}.
    \item $\bm{\Sigma}_n \in \mathbb{R}^{3 \times 3}$, which is the covariance matrix, encoding the anisotropic spatial uncertainty and effective spatial extent of the scatterer. This model shows that real-world scatterers are not point-like but have orientation and spread in space.
  \item $\alpha_{n} \in \mathbb{C}$, which is the intrinsic, distance-dependent complex attenuation coefficient of the $n$-th scatterer, modeling both amplitude attenuation and phase shift imparted by the scatterer. This parameter captures the inherent electromagnetic interaction strength of the scatterer with the radio wave, determined by its material, geometry, and other physical factors. We parameterize $\alpha_n$ as
\begin{align}\label{eq:atten}
\alpha_n = a_n \, e^{j \varphi_n}
\end{align}
where $a_n \geq 0$ is the magnitude attenuation and $\varphi_n \in [-\pi, \pi]$ is the phase shift induced by scatterer $n$. To explicitly model the geometric path loss effects, such as distances from the base station to scatterer, $d_{\mathrm{BS},n}$, both $a_n$ and $\varphi_n$ are treated as the outputs of one encoder-decoder designed MLP as
\begin{align}(a_n,\varphi_n)=\textbf{MLP}\left(\gamma(d_{\mathrm{BS},n});\bm{\Theta}_1\right),
\end{align}
with learnable parameters $\bm{\Theta}_1$ optimized jointly with other scatterer attributes. To enhance spatial resolution, we encode the distance information using a position encoding function inspired by \cite{tancik2020fourier} $\gamma(d_{\mathrm{BS},n})=[\sin(\pi d_{\mathrm{BS},n}),\cdots,\sin(\pi^Vd_{\mathrm{BS},n}),\cos(\pi^V d_{\mathrm{BS},n})]$, where $V$ is the order of the encoding.
  {Physically, $\alpha_n$ serves as an effective composite attenuation factor that holistically encodes overall electromagnetic interaction mechanisms for the $n$-th Gaussian scatterer. Since $\alpha_n$ is geometrically anchored to the Gaussian primitive whose mean $\bm{\mu}_n$ is initialized from the point cloud, the learned attenuation is spatially localized to specific physical structures, enabling generalization to nearby unmeasured locations.}
\item To rigorously account for the directionality dependence, we define $G_n(\theta_{\mathrm{BS},n}, \phi_{\mathrm{BS},n}, \theta_{n,l}, \phi_{n,l})$ as the complex gain function to describe the complex gain caused by the $n$-th virtual scatterers from the BS to scatterer $n$ in departure direction $(\theta_{\mathrm{BS},n}, \phi_{\mathrm{BS},n})$ and from scatterer $n$ to grid $l$ in arrival direction $(\theta_{n,l}, \phi_{n,l})$.
According to the complex spherical harmonic (SH) expansion of a radio wave \cite{1284991}, each gain can be expanded as\footnote{The order of SH expansion $S$ is a hyperparameter and we use $S=4$ in this paper.}
\begin{align}\label{eq:G}
G_n&(\theta_{\mathrm{BS},n}, \varphi_{\mathrm{BS},n}, \theta_{n,l}, \varphi_{n,l})=\nonumber\\
\sum_{s=0}^{S} \sum_{t=-s}^{s} &\tau_{n,st}\, Y_{s t}(\theta_{\mathrm{BS},n}, \varphi_{\mathrm{BS},n})\, Y_{s t}(\theta_{n,l}, \varphi_{n,l}),
\end{align}
where $\tau_{n,st}$ are the intrinsic learnable complex SH coefficients and
\begin{align*}
    Y_{s t}(\theta, \varphi) &= \sqrt{ \frac{2s+1}{4\pi} \frac{(s-t)!}{(s+t)!} } P_s^{|t|}(\cos\theta) e^{it\varphi},
\end{align*}
and $P_s^{|t|}(\cdot)$ is the associated Legendre function \cite{colton1998inverse}.
\end{itemize}

Collectively, the set of attributes for all scatterers is denoted as
\[
\bm{\Gamma}^{PC-TGS} = \{ \bm{\mu}_n, \bm{\Sigma}_n, \alpha_{n}, G_n\}_{n=1}^N.
\]

\subsection{Tangent-Plane Projection}

To render the APS for the AoD from the BS, we must map the 3D Gaussian scatterers to the 2D angular domain as viewed from $\bm{P}_{\mathrm{BS}}$. We present a physically-grounded projection approach leveraging the tangent plane to the unit sphere at each angular bin $(\theta_p, \varphi_q)$, corresponding to the discretized angular bins defined in Section \ref{Subsection: LSCM}. The key idea of the tangent-plane projection is to treat the tangent plane at each angular bin as a virtual receiving plane, as shown in Fig.\ref{fig:projection}. Projecting all the 3D Gaussians onto each tangent plane and aggregating the projected results leads to the channel coefficients in each angular bin and produces the APS. The detailed procedures are as follows:
\begin{figure}
\begin{center}
\centerline{\includegraphics[width=1\columnwidth]{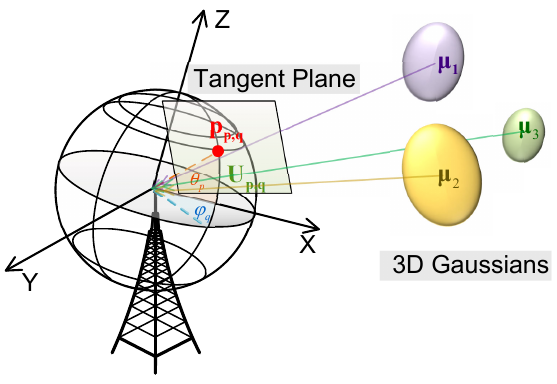}}
\caption{%
Illustration of the tangent-plane projection.
}
\label{fig:projection}
\end{center}
\end{figure} 
\begin{enumerate}
    \item \textbf{Step 1. Construct the tangent plane basis:}  
    With the BS at the origin, the tangent plane basis at angular bin $(\theta_p, \varphi_q)$ is given by
    \begin{equation}\label{eq:defU}
       \bm{U}_{p,q} = [\mathbf{u}_{p,q,1}~\mathbf{u}_{p,q,2}],
    \end{equation}
    where
    \begin{align*}
       \mathbf{u}_{p,q,1} &=
        \begin{bmatrix}
        \cos\theta_p\cos\varphi_q \\
        \cos\theta_p\sin\varphi_q \\
        -\sin\theta_p
        \end{bmatrix}, \quad
        \mathbf{u}_{p,q,2}=
        \begin{bmatrix}
        -\sin\varphi_q \\
        \cos\varphi_q \\
        0
        \end{bmatrix}
    \end{align*}
    are the tangent plane basis vectors.

    \item \textbf{Step 2. Project each 3D Gaussian to the tangent plane:}  
    For each scatterer $n$, the projected mean and covariance on the tangent plane corresponding to $(\theta_p, \varphi_q)$ are
    \begin{equation}
        \begin{aligned}
            \bm{\mu}_{p,q,n} &= \bm{U}_{p,q}^\top (\bm{\mu}_n - \mathbf{p}_{p,q}) \\
            \bm{\Sigma}_{p,q,n} &= \bm{U}_{p,q}^\top \bm{\Sigma}_n \bm{U}_{p,q},
        \end{aligned}
    \end{equation}
    where $\mathbf{p}_{p,q} = \bm{P}_{BS} + \mathbf{n}_{p,q}$ and
    \begin{equation}
        \mathbf{n}_{p,q} =
        \begin{bmatrix}
        \sin\theta_p\cos\varphi_q \\
        \sin\theta_p\sin\varphi_q \\
        \cos\theta_p
        \end{bmatrix}.
    \end{equation}
    
    After projection, each 3D Gaussian becomes a 2D Gaussian on the local tangent plane at $(\theta_p, \varphi_q)$ with density
    \begin{equation}\label{eq:2fgspdf}
        f_{p,q,n}(\mathbf{y}_{p,q}) = \mathcal{N}\left(\mathbf{y}_{p,q};\, \bm{\mu}_{p,q,n},\, \bm{\Sigma}_{p,q,n}\right),
    \end{equation}
    where $\mathbf{y}_{p,q}$ denotes 2D coordinates in the tangent plane.
\end{enumerate}

This 2D Gaussian in the tangent plane forms the basis for rigorous APS synthesis as illustrated in the following subsection.

\subsection{APS Synthesis via Integration-aided Rendering}

To compute the APS for each angular bin, we integrate each projected 2D Gaussian density over the corresponding tangent plane, and aggregate all integration results to obtain the final APS.

\subsubsection{Spherical-Patch Integration}

Let $f_{p,q,n}(\mathbf{y}_{p,q})$ denote the $n$-th projected 2D Gaussian PDF on the tangent plane at $(\theta_p, \varphi_q)$. The corresponding spherical patch (angular bin) is
\begin{equation}
\begin{aligned}\label{eq:regionR}
\mathcal{R} = \Big\{\, (\theta', \varphi') :\;\;
& \theta' \in [\theta_p - \Delta\theta,\, \theta_p + \Delta\theta], \\
& \varphi' \in [\varphi_q - \Delta\varphi,\, \varphi_q + \Delta\varphi]\, \Big\}
\end{aligned}
\end{equation}
Each angle in $\mathcal{R}$ is mapped onto the tangent plane by
\begin{equation}\label{eq:mapg}
\mathbf{g}(\theta', \varphi') = \bm{U}_{p,q}^\top \left( \frac{1}{\mathbf{n}'^\top \mathbf{n}_{p,q}} \mathbf{n}' - \mathbf{n}_{p,q} \right),
\end{equation}
where $\mathbf{n}' = [
\sin\theta' \cos\varphi',\;
\sin\theta' \sin\varphi',\;
\cos\theta'
]^\top$ is the 3D unit direction.

The integration of the $n$-th projected Gaussian on the tangent plane at $(\theta_p, \varphi_q)$ over $\mathcal{R}$ can be expressed as
\begin{equation}\label{eq:exactIntpq}
\bm{w}^{\text{exact}}_{n}(p,q) =
\iint_{\mathcal{R}}
f_{p,q,n}\big( \mathbf{g}(\theta', \varphi') \big)
\left| \det J_{\mathbf{g}}(\theta', \varphi') \right| \;
d\theta' d\varphi',
\end{equation}
where $J_{\mathbf{g}}$ is the $2 \times 2$ Jacobian of $\mathbf{g}(\theta', \varphi')$.

\subsubsection{APS Synthesis}\label{sec:Synthesis}

We now explain how the previously projected 2-D Gaussian footprints are combined into an angular-power-spectrum (APS) estimate at a target grid
$l$.  For each scatterer $n\!\in\![N]$ we already know  
 its complex weight on beam $(p,q)$,
$\bm{w}^{\text{exact}}_{n}(p,q)$ in \eqref{eq:exactIntpq},
its intrinsic attenuation magnitude--phase pair
$\alpha_n$ in \eqref{eq:atten}, and   its direction-dependent gain
$G_n(\theta_{\mathrm{BS},n},\varphi_{\mathrm{BS},n},\theta_{n,l},\varphi_{n,l})$ obtained from \eqref{eq:G}.
A further distance-based factor \footnote{The path-loss factors $\gamma_1$ and $\gamma_2$ are set to be $2$ in this paper.}
$
  L^{(l)}_n
  =
  d_{\mathrm{BS},n}^{-\gamma_1}\,
  d_{n,l}^{-\gamma_2}
$
captures free-space loss along the two-hop path
from the BS to the scatterer and then to the $l$-th grid.
{Because nearer scatterers can partially mask farther ones, we
must decide a front-to-back compositing order before the final synthesis.
Following the radio propagation path, all scatterers are therefore sorted in ascending
depth---i.e., by the distance $d_{n,l}$ of their means
$\bm \mu_{n}$ towards
$\bm{P}_l$.  The resulting permutation
$\pi:[N]\rightarrow[N]$ defines the depth-ordered list  
$\mathbb{I}_{p,q}=\{\pi(1),\pi(2),\dots,\pi(N)\}$ for every beam.}

Given this order, the power received at the $l$-th grid from bin $(\theta_{p},\varphi_{q})$ is

\begin{equation}
\label{eq:complex_alpha_blending}
\begin{aligned}
x^{(l)}_{p,q}=&
\Bigl|
  \sum_{n=1}^{N}
    \alpha_{\pi(n)}
    G_{\pi(n)}(\theta_p,\varphi_q,\theta_{\pi(n),l},\varphi_{\pi(n),l})
    \bm{w}^{\text{exact}}_{\pi(n)}(p,q)
    L^{(l)}_{\pi(n)}\\
    &\times\prod_{m=1}^{n-1}
      \bigl(
        1-
        \alpha_{\pi(m)}
        \bm{w}^{\text{exact}}_{\pi(m)}(p,q)
        L^{(l)}_{\pi(m)}
      \bigr)
\Bigr|^{2}.
\end{aligned}
\end{equation}

Repeating \eqref{eq:complex_alpha_blending} over the full
$N_V\times N_H$ beam grid and concatenating the results yields the APS
vector
$
  \hat{\bm{x}}_l
  =
  [x^{(l)}_{1,1},\dots,x^{(l)}_{N_V,N_H}]^\top.
$

\subsection{Learning Objective and Recursive Fine-Tuning}

The model jointly learns all parameters involved in the APS synthesis and RSRP mapping. The full set of trainable parameters is denoted as
\[
\bm{\Gamma}^{\text{PC-TGS}} = \left\{
    \{\bm{\Sigma}_n,\, \{\tau_{n,st}\}_{0 \leq s \leq S,\; -s \leq t \leq s} \}_{n=1}^N,\;
    \bm{B},\;
    \bm{T},\; \bm{\Theta}_1
\right\},
\]where $ \bm{B}$ and $\bm{T}$ reparameterize  $\{\bm{\mu}_n\}_{n=1}^N$, $\{\tau_{n,st}\}_{0 \leq s \leq S,\; -s \leq t \leq s} \}_{n=1}^N$ reparameterize $\{G_n\}_{n=1}^N$, and $\bm{\Theta}_1$ reparameterizes $\{\alpha_{n}\}_{n=1}^N$.
These parameters are optimized by minimizing a composite loss over the measured grid locations $l \in \mathcal{L}_k$:
\begin{equation}
    \mathcal{L}_{\text{final}} = \underbrace{\sum_{l \in \mathcal{L}_k} \|\bm{y}_l - \bm{A} \hat{\bm{x}}_l\|_2^2
    + \lambda_{\text{TV}}\sum_{l} \|\hat{\bm{x}}_l - \hat{\bm{x}}_{KNN(l)}\|_1}_{ \mathcal{L}_{\text{APS}}}
    + \mathcal{L}_{\text{reg}},
    \label{eq:finalloss}
\end{equation}
where $\bm{y}_l$ is the measured RSRP vector at grid $l \in \mathcal{L}_k$, $\hat{\bm{x}}_l = [x_{1,1}^{(l)}, \ldots, x_{N_V,N_H}^{(l)}]$ is the predicted APS vector for $l$ (see Section~\ref{sec:Synthesis}), $\bm{A}$ is the linear mapping from APS to RSRP, $KNN(l)$ denotes the $K$-nearest neighbors of $l$ in the grid,  $\lambda_{\text{TV}}$ is a hyperparameter controlling the total variation (TV) loss for spatial smoothness, and $\mathcal{L}_{\text{reg}}$ includes regularizers on the model parameters specified in Section ~\ref{sec:RM}. 
This loss is minimized via mini-batch gradient descent, balancing accuracy, spatial coherence, and model regularity. 

To further enhance generalization and extrapolation, we employ a {\em recursive fine-tuning scheme}, inspired by boundary-propagation strategies:
\begin{enumerate}
    \item \textbf{Initial Training:} Optimize the model on the initial training set $\mathcal{L}_k$, where RSRP is measured, using the loss in~\eqref{eq:finalloss}.
    \item \textbf{Boundary Prediction:} Select a set of unmeasured grids near the current training region boundary, denoted $\mathcal{L}_D \subseteq \mathcal{L}_o$ (within $D=4.5$\,grids of distance).
    \item \textbf{Expansion:} Predict APS and RSRP at these boundary grids $\mathcal{L}_D$ using the current model, and augment the training set: $\mathcal{L}_k \leftarrow \mathcal{L}_k \cup \mathcal{L}_D$, $\mathcal{L}_o \leftarrow \mathcal{L}_o \setminus \mathcal{L}_D$.
    \item \textbf{Recursive Refinement:} Retrain or fine-tune the model on the expanded set $\mathcal{L}_k$, repeating steps 2--3 until all unmeasured grids are covered.
\end{enumerate}
This recursive procedure allows the model to iteratively propagate learned multipath and environmental structure from measured regions into previously unseen areas. During testing, hardening of $\bm{T}$ is applied, and the final trained model is used to predict APS  at all unmeasured target locations.
 {\begin{remark}[Justification of Adapting  3DGS to Radio Propagation]
The adaptation of 3DGS to the radio domain rests on two physical
correspondences. First, each 3D Gaussian ellipsoid
$(\bm{\mu}_n, \bm{\Sigma}_n)$ serves as a \emph{virtual scatterer
cluster}: its mean encodes the scatterer location, its anisotropic
covariance captures the effective spatial extent and orientation of the
reflecting or diffracting surface which is consistent with the cluster concept
in geometry-based stochastic channel models, and its
spherical-harmonic gain $G_n$ endows the cluster with an
angle-dependent complex scattering pattern.
Second, the received radio signal is the coherent sum of attenuated multipath components (MPCs).
The depth-ordered accumulation in~\eqref{eq:complex_alpha_blending}
directly realizes this principle in a differentiable pipeline. Each
Gaussian contributes one MPC whose amplitude is governed by the learned
attenuation $\alpha_n$, while the front-to-back ordering naturally
accounts for progressive shadowing and occlusion which can be regarded as the radio
counterpart of volumetric opacity in visual rendering.
\end{remark}}
 {\subsection{GWA-based Integration Approximation}
\label{sec:GWA-method-theory}}

 {The bin integral in \eqref{eq:exactIntpq} generally lacks a closed form due to the nonlinear spherical-to-tangent mapping, complicating deployment. This subsection derives a tractable, closed-form approximation of each bin contribution and quantifies its error by
(i) locally linearizing the spherical-to-tangent map over the bin; and
(ii) replacing the hard-edged parallelogram footprint by a moment-matched Gaussian window (GWA), yielding a closed form.}

 {\subsubsection{Local Linearization}}
For sufficiently small angular bins, we linearize $\mathbf{g}(\theta', \varphi') $ in (\ref{eq:mapg}) at the bin center by
\begin{equation}
\mathbf{g}(\theta', \varphi') \approx \bm{J}_{\mathbf{g}}(\theta_p, \varphi_q)
\begin{bmatrix}
\theta' - \theta_p \\
\varphi' - \varphi_q
\end{bmatrix},
\end{equation}
where $\bm{J}_{\mathbf{g}}(\theta_p, \varphi_q)$ is the Jacobian at the center. Under this linearization, the integration region in the tangent plane becomes a parallelogram, and (\ref{eq:exactIntpq}) can be converted to 
\begin{equation}\label{eq:linParallelogram}
\begin{aligned}
\bm{w}^{\text{lin}}_{n}(p,q) &=
\iint_{\hat{\mathbf{y}}_{p,q} \in \mathcal{R}_{\hat{\mathbf{y}}_{p,q}}}
\mathcal{N}(\hat{\mathbf{y}}_{p,q} ;\, \bm{\mu}_{p,q,n},\, \bm{\Sigma}_{p,q,n})\, d\hat{\mathbf{y}}_{p,q}\\
=&\iint_{\mathbb{R}^2}
\mathcal{N}(\hat{\mathbf{y}}_{p,q} ;\, \bm{\mu}_{p,q,n},\, \bm{\Sigma}_{p,q,n})\, \mathbb{I}_{\mathcal{R}_{\hat{\mathbf{y}}_{p,q}}}( \hat{\mathbf{y}}_{p,q})d\hat{\mathbf{y}}_{p,q}
\end{aligned}
\end{equation}
where $\hat{\mathbf{y}}_{p,q}$ is defined as the linearized tangent-plane coordinate,
\begin{equation}
\hat{\mathbf{y}}_{p,q} = \bm{J}_{\mathbf{g}}(\theta_p, \varphi_q) 
\begin{bmatrix}
\theta' - \theta_p \\ \varphi' - \varphi_q
\end{bmatrix},
\end{equation}
and $\mathcal{R}_{\hat{\mathbf{y}}_{p,q}}$ is the image of the angular bin $\mathcal{R}$ under this linearized mapping, i.e.,
\begin{equation}
\begin{aligned}\label{eq:regionRy}
\mathcal{R}_{\hat{\mathbf{y}}_{p,q}} = \Big\{\, 
  \hat{\mathbf{y}}_{p,q} = \bm{J}_{\mathbf{g}}(\theta_p, \varphi_q)
  \begin{bmatrix}
    \delta\theta \\ \delta\varphi
  \end{bmatrix}\ :\ 
  &\ \delta\theta \in [-\Delta\theta,\, \Delta\theta], \\
  &\ \delta\varphi \in [-\Delta\varphi,\, \Delta\varphi]
\, \Big\}
\end{aligned}
\end{equation}
which is a parallelogram in the tangent plane. However, due to the \emph{hard indicator} $\mathbb{I}_{\mathcal{R}_{\hat{\mathbf{y}}_{p,q}}}( \hat{\mathbf{y}}_{p,q})$, manipulation on  $\bm{w}^{\text{lin}}_{n}(p,q)$ is still intractable. Therefore, we propose the following GWA-based approximation to replace the hard indicator.

 {\subsubsection{GWA-based Approximation}}

The hard indicator in~(\ref{eq:linParallelogram}) can be viewed as an unnormalized uniform density over the angular bin as it assigns equal weight to all points within the bin and zero elsewhere, physically aggregating scatterer energy over each discrete direction. To enable closed-form analysis, we approximate this uniform distribution with a Gaussian window that matches both its area and second central moments. This moment-matched GWA  preserves the total energy and spatial spread of the bin, providing a computationally efficient and physically faithful surrogate for the original integration.

Concretely, we replace $\mathbb{I}_{\mathcal{R}_{\hat{\mathbf{y}}_{p,q}}}( \hat{\mathbf{y}}_{p,q})$ with a Gaussian window,
\begin{align}
W(\hat{\mathbf{y}}_{p,q}) = A_{\hat{\mathbf{y}}_{p,q}}\,\mathcal{N}\big(\hat{\mathbf{y}}_{p,q};\,\mathbf{0},\,\mathbf{S}_{\hat{\mathbf{y}}_{p,q}}\big),
\end{align}
where the normalization and covariance are given by
\begin{align}
A_{\hat{\mathbf{y}}_{p,q}} &= \mathrm{Area}({\mathcal{R}_{\hat{\mathbf{y}}_{p,q}}}) = 4\,\Delta\theta\,\Delta\varphi\,\bigl|\det \bm{J}_{\mathbf{g}}(\theta_p, \varphi_q)\bigr|, \\
\mathbf{S}_{\hat{\mathbf{y}}_{p,q}} &= \bm{J}_{\mathbf{g}}(\theta_p, \varphi_q)\,\mathbf{S}_{\mathrm{ang}}\, \bm{J}_{\mathbf{g}}(\theta_p, \varphi_q)^\top, \\
\mathbf{S}_{\mathrm{ang}} &= \begin{bmatrix}\Delta\theta^2/3 & 0\\ 0 & \Delta\varphi^2/3\end{bmatrix}.
\end{align}
Substituting $W(\hat{\mathbf{y}}_{p,q})$ into~(\ref{eq:linParallelogram}) yields
\begin{equation}
\begin{aligned}
&\bm{\widetilde{w}}_n(p,q)
= \int_{\mathbb{R}^2} f_{p,q,n}(\hat{\mathbf{y}}_{p,q})\,W(\hat{\mathbf{y}}_{p,q})\,\mathrm{d}\hat{\mathbf{y}}_{p,q} \\
&= A_{\hat{\mathbf{y}}_{p,q}} \int_{\mathbb{R}^2} \mathcal{N}(\hat{\mathbf{y}}_{p,q};\boldsymbol{\mu}_{p,q,n},\boldsymbol{\Sigma}_{p,q,n}) \,
\mathcal{N}(\hat{\mathbf{y}}_{p,q};\mathbf{0},\mathbf{S}_{\hat{\mathbf{y}}_{p,q}})\,\mathrm{d}\hat{\mathbf{y}}_{p,q},
\end{aligned}
\end{equation}
which, by Gaussian convolution, simplifies to the closed-form
\begin{equation}\label{eq:GWAApproxFinal}
\bm{\widetilde{w}}_n(p,q)
= A_{\hat{\mathbf{y}}_{p,q}}\, \mathcal{N}\big(\mathbf{0};\, \boldsymbol{\mu}_{p,q,n},\, \boldsymbol{\Sigma}_{p,q,n}+\mathbf{S}_{\hat{\mathbf{y}}_{p,q}}\big).
\end{equation}
Replacing $ \bm{w}^{\text{exact}}_n(p,q)$ in~(\ref{eq:complex_alpha_blending}) by $\bm{\widetilde{w}}_n(p,q)$  yields an efficient and accurate APS synthesis.

 {\subsubsection{Error Analysis}}
\label{sec:error-analysis}

We analyze the accuracy of the GWA-based approximation in \eqref{eq:GWAApproxFinal} relative to the exact bin integral \eqref{eq:exactIntpq}. Before stating the main error bound, we establish two key regularity results. 

\begin{lemma}[Regularity of the spherical-to-tangent mapping]\label{lem:g-regularity-inline}
    Let $\boldsymbol{\delta} \triangleq (\theta'-\theta_p,\,\varphi'-\varphi_q)^\top$
    denote the angular displacement from the bin center $(\theta_p,\varphi_q)$.
    There exists a neighborhood $\mathcal{R}$ of $(\theta_p,\varphi_q)$
    and finite positive constants $C_J$ and $C_g$, depending only on
    $\mathcal{R}$, such that $\mathbf{g}(\theta',\varphi')$ from
    \eqref{eq:mapg} is infinitely differentiable ($C^\infty$) in
    $\mathcal{R}$, and its Jacobian and Hessian are uniformly bounded:
\begin{equation}
    \sup_{(\theta',\varphi')\in\mathcal{R}}\|\bm{J}_{\mathbf{g}}(\theta',\varphi')\|\le C_J, \qquad
    \sup_{(\theta',\varphi')\in\mathcal{R}}\|D^2\mathbf{g}(\theta',\varphi')\|\le C_g.
\end{equation}
Moreover, for all $(\theta',\varphi') \in \mathcal{R}$, the Taylor expansion up to second order holds:
\[
\mathbf{g}(\theta',\varphi')
=\mathbf{g}(\theta_p,\varphi_q)+\bm{J}_{\mathbf{g}}(\theta_p,\varphi_q)\boldsymbol{\delta}
+\tfrac{1}{2}\mathbf{Q}(\boldsymbol{\delta})
+\mathbf{R}_{\mathrm g}(\boldsymbol{\delta}),
\]
with $\|\mathbf{Q}(\boldsymbol{\delta})\|\le C_g\|\boldsymbol{\delta}\|^2$ and $\|\mathbf{R}_{\mathrm g}(\boldsymbol{\delta})\|\le C_g\|\boldsymbol{\delta}\|^3$.
\end{lemma}
\begin{proof}
  Proof in Appendix~\ref{app:lemmas1}.
\end{proof}

\begin{lemma}[Regularity of the integration region]\label{lem:local-invertibility-inline}
   Suppose the Jacobian $\bm{J}_{\mathbf{g}}(\theta_p,\varphi_q)$ at $(\theta_p, \varphi_q)$ is invertible. Then there exists $r > 0$ such that, for any rectangle $\mathcal{R}$ as in (\ref{eq:regionR}) with $\max\{\Delta\theta,\Delta\varphi\}<r$, the mapping $\mathbf{g}(\theta',\varphi')$ is a $C^1$ diffeomorphism from $\mathcal{R}$ onto its image $\mathbf{g}(\mathcal{R})$, and $|\det\bm{J}_{\mathbf{g}}(\theta, \varphi)|$ is bounded below by a positive constant for all $(\theta, \varphi) \in \mathcal{R}$. The linear image $\mathcal{R}_{\hat{\mathbf{y}}_{p,q}}$ in (\ref{eq:regionRy}) is a nondegenerate parallelogram in the tangent plane.
\end{lemma}

\begin{proof}
    Proof in Appendix~\ref{app:lemmas3}.
\end{proof}

With these technical ingredients, we state the main result:

\begin{theorem}[Approximation error of tangent-linearized GWA bin integration]
\label{thm:GWA-error-full}
There exist constants $C_3,C_4>0$ such that
\begin{equation}
\label{eq:main-bound}
\begin{aligned}
    \big|\,\bm{w}^{\text{exact}}_{n}(p,q) - \bm{\widetilde{w}}_n(p,q)\,\big|
\le& C_3\!\left(\Delta\theta^3+\Delta\varphi^3\right)\\
&+ C_4\!\left(\Delta\theta^4+\Delta\varphi^4+\Delta\theta^2\Delta\varphi^2\right).
\end{aligned}
\end{equation} Here, $\bm{w}^{\text{exact}}_{n}(p,q)$ is the exact bin integral \eqref{eq:exactIntpq}, and $\bm{\widetilde{w}}_n(p,q)$ is the GWA approximation \eqref{eq:GWAApproxFinal}.
\end{theorem}
\begin{proof}
    Proof in Appendix~\ref{app:thm1}.
\end{proof}
 {\begin{remark}[Numerical Validation of Theorem~\ref{thm:GWA-error-full}]
We validate the error bound with the following settings:
$\bm{P}_{BS}=\mathbf{0}$, $\bm{\Sigma}_n=\sigma^2\mathbf{I}_3$
($\sigma=0.15$\,rad),
$\bm{\mu}_n=[1.2,\,0.030,\,-0.045]^\top$, and
$\Delta\theta=\Delta\varphi=\Delta$.
Fig.~\ref{fig:gwa_curves} plots the absolute error
$|\bm{w}^{\text{exact}}_{n}(p,q) - \bm{\widetilde{w}}_n(p,q)|$
vs.\ $\Delta$ on log--log axes at the equatorial bin
$(\theta_p,\varphi_q)=(\pi/2,0)$. The result shows that the empirical slope ($\approx3.9$)
exceeds the theoretical slope of~3, confirming the
$O(\Delta\theta^3+\Delta\varphi^3)$ bound and showing the GWA
is empirically tighter than guaranteed.
\end{remark}}

 {\begin{figure}[ht]
\centering
\includegraphics[width=0.9\columnwidth]{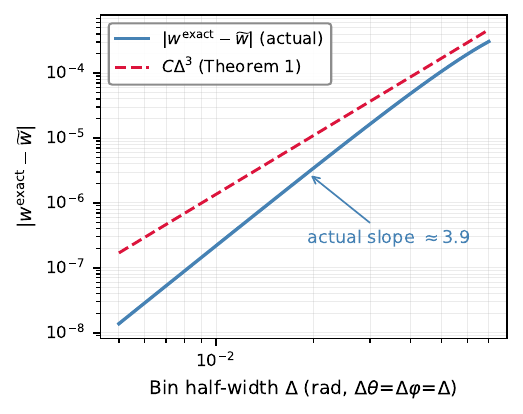}
\caption{Numerical validation of Theorem~\ref{thm:GWA-error-full}:
log--log plot of $|w^{\text{exact}}_{n}(p,q)-\widetilde{w}_n(p,q)|$
vs.\ bin half-width $\Delta$ ($\Delta\theta=\Delta\varphi=\Delta$)
at the equatorial bin. The empirical slope ($\approx3.9$) exceeds the
theoretical cubic bound slope of~3.}
\label{fig:gwa_curves}
\end{figure}}
This error analysis provides rigorous, quantitative guarantees for the GWA bin integration method, ensuring that the approximation is reliable for practical bin widths.

\section{Experiment Results and Discussion}\label{sec:experiments}

\subsection{Experimental Setup}\label{sec:exp-setup}

We evaluate PC-TGS on a city-scale, multi-modal dataset collected in City~A, China. The dataset contains over five million LiDAR points and $|\mathcal{L}|=6{,}310$ beam-wise RSRP samples within a $1{,}000\times1{,}200\,\text{m}^2$ area. The area is discretized into $10\times10\,\text{m}^2$ grids at a receiver height of $1\,\text{m}$. The base station (BS) transmits $M=32$ reference beams. We collect two rounds of measurements from the same BS with different parameter settings:
\begin{itemize}
    \item \textbf{Round 1 (training configuration):} RSRP vectors $\{\bm y_\ell \in \mathbb{R}^{M}\mid \ell\in\mathcal{L}\}$ under measurement matrix $\bm A \in \mathbb{R}^{M\times N_A}$.
    \item \textbf{Round 2 (rotated configuration):} RSRP vectors $\{\bm y^r_\ell \in \mathbb{R}^{M}\mid \ell\in\mathcal{L}\}$ under $\bm A^{r} \in \mathbb{R}^{M\times N_A}$.
\end{itemize}

Model training uses only a subset of Round-1 RSRPs from measured grids $\mathcal{L}_k$ with $|\mathcal{L}_k|=5{,}566$, as shown in the red region of  Fig.~\ref{fig:Task}. 
Unless otherwise stated, we discretize the angular domain into $N_A=N_V\times N_H=72\times 91=6{,}552$ bins.  {All experiments are carried out on an NVIDIA A100 40\,GB GPU. We implement PC-TGS in PyTorch~2.8.0 and Python~3.13.5. PC-TGS uses $N=2{,}000$ virtual scatterers and is trained with the Adam optimizer~\cite{adam2014method} using the loss in~\eqref{eq:finalloss}.} All baselines are retrained on our dataset and evaluated on the same test splits.

\subsection{Baseline Methods and Evaluation Metrics}

 {We benchmark PC-TGS against state-of-the-art methods and ablated variants. Methods are categorized as \emph{Offline} (require a training phase, then fast inference) or \emph{Online} (solved per-query at inference time, no training):}
\begin{itemize}
    \item \textbf{WNOMP}~\cite{LSCM}  {\emph{(Online)}}: An orthogonal matching pursuit-based method for recovering APS from RSRP. For unmeasured regions, it is combined with Kriging~\cite{sato2020radio}.
    \item  {\textbf{TMSBL}~\cite{zhang2011sparse} \emph{(Online)}: Temporally correlated Multiple Sparse Bayesian Learning, which groups spatially or signal-similar samples to exploit joint sparsity for APS recovery via a learned cross-measurement covariance structure.}
    \item  {\textbf{Ray Tracing}~\cite{eertmans2025differt} \emph{(Online)}: A deterministic physics-based propagation simulator using differentiable ray tracing to predict APS from geometric scene models.}
    \item  {\textbf{RadioUNet}~\cite{levie2021radiounet} \emph{(Offline)}: A CNN-based deep learning model for radio map prediction from sparse RSRP measurements and building geometry.}
    \item \textbf{MM-LSCM}~\cite{wang2025multi}  {\emph{(Offline)}}: A dual-branch neural radio radiance field multi-modal model integrating RSRP and LiDAR data for the LSCP task.
    \item \textbf{SM-LSCM}~\cite{wang2025multi}  {\emph{(Offline)}}: A single-modal variant of MM-LSCM using only RSRP data, without LiDAR integration.
   \item  {\textbf{PC-TGS} \emph{(Offline)}: The proposed framework incorporating the recursive fine-tuning (RFT) scheme, serving as the default model for best extrapolation accuracy.}
    \item  {\textbf{PC-TGS w/o RFT} \emph{(Offline)}: A lighter variant of PC-TGS that omits the recursive fine-tuning scheme, offering reduced training time with a favorable accuracy--efficiency trade-off.}
    \item  {\textbf{SM-PC-TGS} \emph{(Offline)}: A single-modal variant in which the scatterer means $\{\bm{\mu}_n\}_n^N$ are initialized from a uniform distribution without using point cloud data, and the RM scheme is omitted.}
    \item  {\textbf{SM-PC-TGS w/o RFT} \emph{(Offline)}: A single-modal variant in which the scatterer means $\{\bm{\mu}_n\}_n^N$ are initialized from a uniform distribution without using point cloud data, and the RM scheme and RFT are both omitted.}
\end{itemize}

We evaluate the performance of the proposed PC-TGS  approach across three prediction tasks:
\begin{itemize}
    \item \textbf{Sub-Task 1 (ST-1):} Predicting rotated RSRPs $ \hat{\bm y}^r_l=\bm A^r\hat{\bm x}_l$ in measured regions $\{l\in\mathcal{L}_k\}$  with $|\mathcal{L}_k|=5{,}566$.
    \item \textbf{Sub-Task 2 (ST-2):} Predicting RSRPs $ \hat{\bm y}_l=\bm A\hat{\bm x}_l$ in unmeasured regions $\{l\in \mathcal{L}_o\}$ with  $|\mathcal{L}_o|=744$.
    \item \textbf{Sub-Task 3 (ST-3):} Predicting rotated RSRPs $\hat{\bm y}^r_l=\bm A^r\hat{\bm x}_l$ in unmeasured regions $\{l\in\mathcal{L}_o\}$.
\end{itemize}
Note that the original training set only includes RSRPs data on the measured region $\{\bm y_l|l\in \mathcal{L}_k\}$.
Performance is evaluated using mean absolute error (MAE) for RSRP prediction as
\begin{align}
    \text{MAE}&=\frac{1}{M}\frac{1}{| \mathcal{L}_o|}\sum_{m=1}^M \sum_{l\in \mathcal{L}_o} |10\log(\hat{\bm{y}}_{m,l}) - 10\log({\bm{y}}_{m,l})|,
\end{align}
where $| \mathcal{L}_o|$ is the total number of grids in the test dataset and
$\hat {\bm{y}}_{l}$ denotes the predicted RSRP at the $l$-th grid.

\subsection{Ablation and Sensitivity Analysis}\label{sec:ablation}

 {Before comparing PC-TGS to external baselines, we first validate the internal design choices and characterize the model's sensitivity to key parameters.}

\subsubsection{Physical Plausibility of Learned Scatterers}

Fig.~\ref{fig:aps-depth} overlays the estimated scatterer locations on the raw point cloud of City A. The 3D Gaussians exhibit strong alignment with prominent environmental structures (building facades, corners), confirming that the relaxed-mean reparameterization successfully anchors scatterers to physically meaningful locations.

\begin{figure}[ht]
    \centering
    \includegraphics[width=0.35\textwidth]{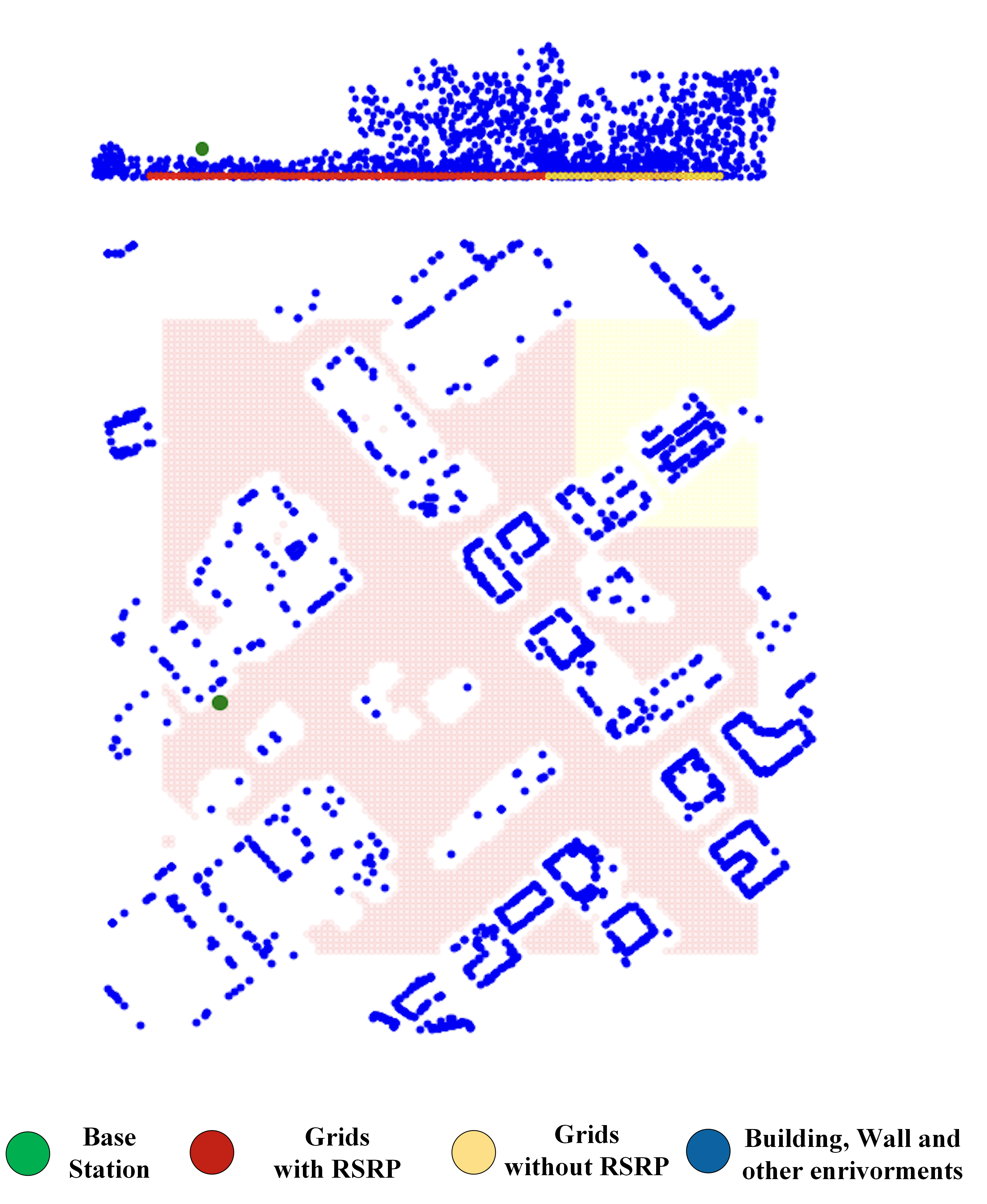}
    \caption{Visualization of extracted virtual scatterers by PC-TGS in City A, China.}
    \label{fig:aps-depth}
\end{figure}

 {\subsubsection{Sensitivity to the Number of Virtual Scatterers}
\begin{table}[htbp]
\centering
\caption{Sensitivity to  the Number of Virtual Scatterers}
\label{tab:scatterers}
\begin{tabular}{lccccc}
\toprule
$N$ & 500 & 1000 & 1500 & 2000 & 3000 \\
\midrule
ST-1 (dB) & 6.56 & 6.23 & 5.78 & 5.57 & 5.85 \\
ST-2 (dB) & 7.99 & 7.76 & 7.64 & 7.45 & 11.23 \\
ST-3 (dB) & 8.13 & 8.02 & 7.82 & 7.57 & 11.58 \\
Time (ms) & 28.91 & 47.99 & 67.42 & 73.55 & 129.29 \\
\bottomrule
\end{tabular}
\end{table}
As shown in Table~\ref{tab:scatterers}, performance improves as $N$ increases from 500 to 2000 and stabilizes in the range 1000--2000 (e.g., Sub-Task~2: 7.99\,dB at $N{=}500$ vs.\ 7.45\,dB at $N{=}2000$). Further increasing to $N{=}3000$ degrades accuracy and raises inference time significantly (from $73.55\,\mathrm{ms}$ to $ 129.29\,\mathrm{ms}$), indicating that $N{=}2000$ captures the dominant propagation structure. Crucially, $N$ controls the model's representational capacity rather than scaling with the geographic size of the environment. For larger environments, a patch-wise deployment is recommended.}

 {\subsubsection{Sensitivity to LiDAR Point Density and Geometric Noise}
\begin{table}[htbp]
\centering
\caption{Sensitivity to LiDAR Sparsification}
\label{tab:lidar_sparsify}
\begin{tabular}{lccc}
\toprule
LiDAR Setting & ST-1 (dB) & ST-2 (dB) & ST-3 (dB) \\
\midrule
100\% (Full) & 5.57 & 7.45 & 7.57 \\
50\% (Random) & 5.64 & 7.84 & 8.00 \\
None (SM-PC-TGS) & 5.84 & 8.44 & 8.79 \\
\bottomrule
\end{tabular}
\end{table}
As shown in Table~\ref{tab:lidar_sparsify}, randomly halving the point cloud causes only modest degradation, indicating that coarse geometric priors suffice. Removing LiDAR entirely (SM-PC-TGS) leads to a larger drop, confirming the value of even sparse geometric information.

To assess compatibility with image-based reconstruction (e.g., DUSt3R~\cite{wang2024dust3r} with typical outdoor accuracy 0.1--0.5\,m vs.\ LiDAR at 0.01--0.05\,m), we simulate reconstruction noise by adding zero-mean Gaussian noise at three levels:
\begin{table}[htbp]
\centering
\caption{Sensitivity to the Point-Cloud Geometric Noise}
\label{tab:pc_noise}
\begin{tabular}{lccc}
\toprule
Noise $\sigma$ & ST-1 (dB) & ST-2 (dB) & ST-3 (dB) \\
\midrule
0 (LiDAR) & 5.57 & 7.45 & 7.57 \\
0.1\,m & 5.61 & 7.43 & 7.83 \\
0.5\,m & 5.67 & 7.70 & 8.21 \\
1.0\,m & 5.79 & 8.05 & 8.45 \\
\bottomrule
\end{tabular}
\end{table}
Table~\ref{tab:pc_noise} shows that PC-TGS remains robust across all noise levels, supporting the practical use of image-derived point clouds as a LiDAR alternative.}

 {\subsubsection{Sensitivity to Angular Resolution}
\begin{table}[htbp]
\centering
\caption{Sensitivity to the Angular Resolution}
\label{tab:angular}
\begin{tabular}{lccc}
\toprule
Resolution & $72{\times}91$ & $36{\times}91$ & $72{\times}45$ \\
\midrule
ST-1 (dB) & 5.57 & 6.37 & 6.48 \\
ST-2 (dB) & 7.45 & 8.39 & 8.57 \\
ST-3 (dB) & 7.57 & 8.55 & 8.68 \\
\bottomrule
\end{tabular}
\end{table}
As shown in Table~\ref{tab:angular}, halving one angular dimension causes only moderate degradation, demonstrating that the geometric 3D representation provides robustness even under substantially coarser angular discretization.}

 {\subsubsection{Sensitivity to RSRP Measurement Density}
\begin{figure}[ht]
\centering
\begin{subfigure}{0.22\textwidth}
\centering
\includegraphics[width=\textwidth]{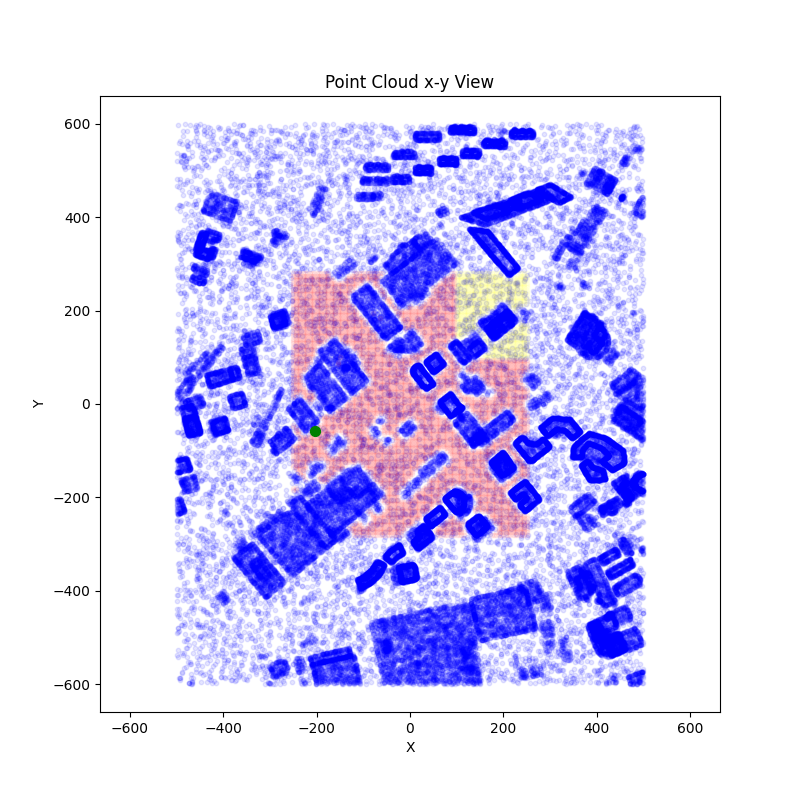}
\caption{88\% (Origin)}
\end{subfigure}
\hfill
\begin{subfigure}{0.22\textwidth}
\centering
\includegraphics[width=\textwidth]{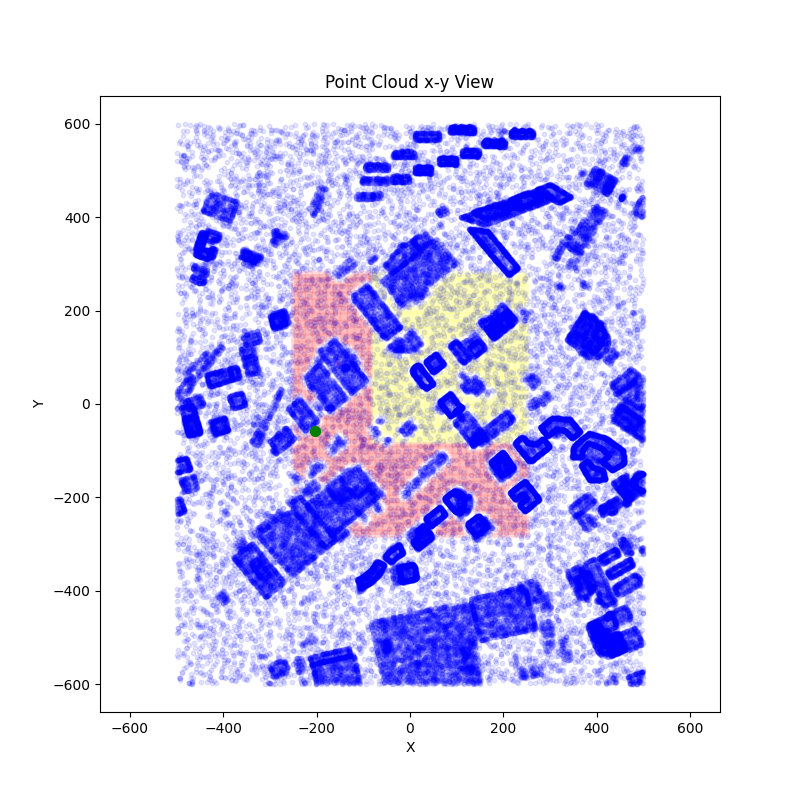}
\caption{50\%}
\end{subfigure}
\caption{Training/test data splits. Red: training; yellow: test.}
\label{fig:data_ratio}
\end{figure}
\begin{table}[htbp]
\centering
\caption{Sensitivity to  the  Training Data Ratio}
\label{tab:data_ratio}
\begin{tabular}{lcc}
\toprule
Training Ratio & 88\% (Origin) & 50\% \\
\midrule
ST-1 (dB) & 5.57 & 5.79 \\
ST-2 (dB) & 7.45 & 9.25 \\
ST-3 (dB) & 7.57 & 9.67 \\
\bottomrule
\end{tabular}
\end{table}
As shown in Fig.~\ref{fig:data_ratio} and Table~\ref{tab:data_ratio}, reducing the training data ratio from 88\% to 50\% increases MAE. However, PC-TGS at 50\% measurement density ($9.25$\,dB) still outperforms RadioUNet trained at full density ($10.31$\,dB, Table~\ref{tab:1}), demonstrating that geometry-informed 3D representation substantially reduces the measurement burden.}

 {\subsubsection{RFT Effectiveness}
As shown in Table~\ref{tab:1}, the asymmetric gains of PC-TGS over PC-TGS w/o RFT---improvements of $0.78$\,dB in Sub-Task~2 (from $8.23$ to $7.45$\,dB) and $1.05$\,dB in Sub-Task~3 (from $8.62$ to $7.57$\,dB), with only a $+0.14$\,dB change in Sub-Task~1 (from $5.43$ to $5.57$\,dB)---confirm that the boundary-expansion fine-tuning primarily benefits extrapolation while preserving performance on the original training region. This asymmetry is the expected signature of controlled boundary expansion and is inconsistent with runaway error drift.}

\subsection{Performance Comparison}\label{sec:perf}

Table~\ref{tab:1} summarizes the performance of  {PC-TGS} and all baselines in three prediction tasks where the training set contains only RSRP data from the explored region $\{\bm{y}_l\mid l\in\mathcal{L}_k\}$.

\begin{table*}[htbp]
\centering
\caption{Performance and Complexity Comparison}
\resizebox{\textwidth}{!}{%
\begin{tabular}{lccccccccc}
\toprule
\hline
\multirow{2}{*}{\textbf{Method}} & \multirow{2}{*}{\textbf{Type}} & \multirow{2}{*}{\textbf{Training Time}} & \multicolumn{2}{c}{\textbf{Sub-Task 1}} & \multicolumn{2}{c}{\textbf{Sub-Task 2}} & \multicolumn{2}{c}{\textbf{Sub-Task 3}} \\ \cline{4-9}
 & & & MAE (dB) & Infer. Time & MAE (dB) & Infer. Time & MAE (dB) & Infer. Time \\ \hline
 {\textbf{PC-TGS}} &  {Offline} &  {3.55 h} &  {5.57} &  {73.97 ms} &  {\textbf{7.45}} &  {73.55 ms} &  {\textbf{7.57}} &  {73.69 ms} \\
 {PC-TGS w/o RFT}          &  {Offline} &  {3.03 h} &  {\textbf{5.43}} &  {73.81 ms} &  {8.23} &  {73.54 ms} &  {8.62} &  {73.60 ms} \\
 {SM-PC-TGS w/o RFT}       &  {Offline} &  {2.96 h} &  {5.75} &  {73.50 ms} &  {9.18} &  {73.42 ms} &  {10.04} &  {72.66 ms} \\
 {MM-LSCM}~\cite{wang2025multi} &  {Offline} &  {4.21 h} &  {5.76} &  {384.40 ms} &  {9.12} &  {379.43 ms} &  {9.94} &  {379.38 ms} \\
 {SM-LSCM}~\cite{wang2025multi} &  {Offline} &  {3.31 h} &  {5.89} &  {98.96 ms} &  {10.02} &  {99.24 ms} &  {11.14} &  {97.39 ms} \\
 {RadioUNet}~\cite{levie2021radiounet} &  {Offline} &  {0.8 h} &  {--} &  {--} &  {10.31} &  {2.08 ms} &  {--} &  {--} \\
 {WNOMP}~\cite{LSCM} &  {Online} &  {--} &  {6.12} &  {1.49 ms} &  {13.35} &  {8031.95 ms} &  {15.07} &  {8016.62 ms} \\
 {TMSBL}~\cite{zhang2011sparse} &  {Online} &  {--} &  {6.03} &  {133.71 ms} &  {11.69} &  {9{,}743.23 ms} &  {11.77} &  {9{,}729.18 ms} \\
 {Ray Tracing}~\cite{eertmans2025differt} &  {Online} &  {--} &  {--} &  {--} &  {14.37} &  {8099.93 ms} &  {16.77} &  {7521.13 ms} \\
\hline\bottomrule
\end{tabular}%
}
\label{tab:1}
\end{table*}\footnotetext{{Ray
Tracing is included in Sub-Tasks 2 and 3 only, as Sub-Task 1 shares the
training footprint and this baseline cannot be meaningfully evaluated there.
RadioUNet cannot handle the task where the APS needs to be recovered from
the RSRP, hence is not included in Sub-Tasks 1 and 3.}}

 {The results in Table~\ref{tab:1} demonstrate that PC-TGS achieves the lowest MAE in the extrapolation sub-tasks (7.45\,dB in Sub-Task~2 and 7.57\,dB in Sub-Task~3), while PC-TGS w/o RFT achieves the best performance in Sub-Task~1 (5.43\,dB). This reflects that the recursive fine-tuning primarily benefits extrapolation. Among offline-trained baselines, MM-LSCM (9.12/9.94\,dB) and SM-LSCM (10.02/11.14\,dB) are outperformed by a large margin in Sub-Tasks~2--3, attributable to PC-TGS's geometric initialization via point clouds and 3DGS pipeline which ensure physically consistent scatterer placement and radio propagation mechanism. RadioUNet (10.31\,dB, Sub-Task~2), despite its CNN-based design, cannot leverage 3D propagation geometry and is limited to predicting scalar RSRP rather than the full APS. Among online methods, WNOMP (13.35/15.07\,dB), TMSBL (11.69/11.77\,dB), and Ray Tracing (14.37/16.77\,dB) are substantially worse in the extrapolation sub-tasks. We can infer that  WNOMP and TMSBL cannot predict beyond measured locations, and Ray Tracing suffers from assumed material parameters that diverge from true propagation physics. These results confirm that PC-TGS achieves superior accuracy by jointly learning a geometry-aware 3D propagation representation from measurements and environmental priors.}

 {\subsection{Complexity Analysis}\label{sec:complexity}}

 {As shown in Table~\ref{tab:1}, PC-TGS requires 3.55 hours of offline training (including the RFT phase), after which inference takes around $74$\,ms per query, approximately  representing a $5\times$ speedup over MM-LSCM and $1.3\times$ over SM-LSCM. While WNOMP has negligible per-query cost in Sub-Task~1 ($1.49$\,ms), it requires $>8$\,s per query when applied to extrapolation via Kriging in Sub-Tasks~2--3. TMSBL and Ray Tracing similarly incur 7.5--10\,s per query. The lighter variant PC-TGS w/o RFT reduces training to 3.03 hours by omitting the RFT phase while keeping inference time unchanged. These results confirm that PC-TGS achieves a favorable accuracy--efficiency trade-off among all methods that can address the LSCP task.}

\section{Conclusion}\label{sec:conclusion}

We proposed \emph{point-cloud-assisted tangent Gaussian splatting}
(PC-TGS), a physically grounded framework that extrapolates the channel APS to grids where no signal measurements are available.
PC-TGS fuses LiDAR geometry and RSRP data in an end-to-end 
differentiable pipeline integrating relaxed-mean scatterer parameterization,
tangent-plane projection, and depth-aware electromagnetic synthesis. This is made possible by a novel 
GWA-based closed-form bin integration.  Experiments on a city-scale data
set show clear accuracy gains over existing single- and
multi-modal baselines.   {Future work will extend PC-TGS to multi-site
scenarios, incorporate temporal dynamics, pursue further runtime
optimisation for real-time digital-twin applications, and investigate
adaptive mechanisms for selecting the number of virtual scatterers $N$
based on scene complexity.}

\appendix

\subsection{Proof of Lemma~\ref{lem:g-regularity-inline}}
\label{app:lemmas1}

For any $(\theta',\varphi') \in \mathcal{R}$ the geodesic distance $\gamma$ between $\bm{n}'$ and $\bm{n}_0$
satisfies $\gamma \le \rho$. We assume the bin is sufficiently small so that $\rho < \pi/2$. Hence, we have
\[
s(\theta',\varphi') \;=\; \bm{n}'^{\top}\bm{n}_0 \;=\; \cos\gamma \;\ge\; \cos\rho \;=:\; c_0 \;>\; 0.
\]
This gives a uniform positive lower bound on $s$ over $\mathcal{R}$. Therefore,  the quotient $\bm{h}=\bm{n}'/s - \bm{n}_0$ is smooth on $\mathcal{R}$. Since  that the orthogonality of $\bm{U}_{p,q}$ and linear projection by $\bm{P}$ preserve smoothness,  $\bm{g}$ is smooth on $\mathcal{R}$.

Next, we bound the first-order derivative at any $i\in\{\theta,\varphi\}$ by
\begin{equation}\label{eq:grad-h-bound}
\|\partial_i \bm{h}\|
\le \frac{\|\partial_i \bm{n}'\|}{c_0} + \frac{|\partial_i s|}{c_0^2}\le \frac{\|\partial_i \bm{n}'\|}{c_0} + \frac{\|\partial_i \bm{n}'\|}{c_0^2}
\le \frac{1}{c_0} + \frac{1}{c_0^2}.
\end{equation}
Let $\nabla \bm{h} = [\partial_\theta \bm{h},\ \partial_\varphi \bm{h}] \in \mathbb{R}^{3\times 2}$ be the Jacobian.
Since $\bm{g} = \bm{P}\bm{U}_{p,q}^{\!\top}\bm{h}$, we have $\bm{J}_{\bm{g}} = \bm{P}\bm{U}_{p,q}^{\!\top}\nabla \bm{h}$.
Because $\bm{U}_{p,q}$ is orthogonal, $\|\bm{U}_{p,q}^{\!\top}\|_2=1$, and since $\bm{P}$ is a row-selector,
$\|\bm{P}\|_2\le 1$. Hence
\begin{equation}\label{eq:Jg-le-grad}
\|\bm{J}_{\bm{g}}\|_2 \le \|\nabla \bm{h}\|_2.
\end{equation}
A conservative bound on the operator norm of a $3\times 2$ matrix is the sum of its column norms:
\begin{equation}
    \begin{aligned}
        \|\nabla \bm{h}\|_2 &
\le \left(\frac{1}{c_0}+\frac{1}{c_0^2}\right) + \left(\frac{1}{c_0}+\frac{1}{c_0^2}\right)
=  C_J.
    \end{aligned}
\end{equation}

Combining with \eqref{eq:Jg-le-grad} yields
\[
\sup_{(\theta',\varphi')\in\mathcal{R}}\|\bm{J}_{\bm{g}}(\theta',\varphi')\|
\le \sup_{(\theta',\varphi')\in\mathcal{R}}\|\nabla \bm{h}(\theta',\varphi')\|
\le C_J.
\]

The bound for the second derivative follows the same fashion,  therefore, we omit it due to the page limit.

Finally, we show the second-order Taylor expansion.
Let $\bm{\delta} := [\delta\theta,\ \delta\varphi]^\top$ and $(\theta',\varphi') = (\theta_p,\varphi_q)+\bm{\delta} $.
By Taylor's theorem for vector-valued functions,
\[
\bm{g}(\theta',\varphi') = \bm{g}_0 + \bm{J}_0\,\bm{\delta} + \tfrac{1}{2}\,\mathcal{Q}_0(\bm{\delta}) + \bm{R}(\bm{\delta}),
\]
where $\bm{g}_0 := \bm{g}(\theta_p,\varphi_q)$, $\bm{J}_0 := \bm{J}_{\bm{g}}(\theta_p,\varphi_q)$,
$\mathcal{Q}_0$ is the bilinear form induced by $D^2 \bm{g}(\theta_p,\varphi_q)$, and $\bm{R}(\delta)$ is the high-order  remainder term. 
Because $\bm{g}$ is $C^\infty$ on a compact neighborhood of $\mathcal{R}$, we have
\begin{equation}
    \begin{aligned}
        &\|\mathcal{Q}_0(\bm{\delta})\| \le \|D^2 \bm{g}(\theta_p,\varphi_q)\|\,\|\bm{\delta}\|^2 \le C_g \|\bm{\delta}\|^2,\\
&\|\bm{R}(\bm{\delta})\| \le C_g \|\bm{\delta}\|^3.
    \end{aligned}
\end{equation}
This finishes the proof.

\subsection{Proof of Lemma~\ref{lem:local-invertibility-inline}}
\label{app:lemmas3}

Let us fix the bin center $(\theta_p, \varphi_q)$ and set $\bm{J}_0 := \bm{J}_{\mathbf{g}}(\theta_p, \varphi_q)$. By assumption, $\bm{J}_0$ is invertible, so $\det\bm{J}_0 \neq 0$.

Since $\mathbf{g}$ is $C^1$ in a neighborhood of $(\theta_p, \varphi_q)$ (by Lemma~\ref{lem:g-regularity-inline}), its Jacobian $\bm{J}_{\mathbf{g}}(\theta, \varphi)$ and determinant $\det\bm{J}_{\mathbf{g}}(\theta, \varphi)$ are continuous functions of $(\theta, \varphi)$. Therefore, there exists $r > 0$ such that for all $(\theta, \varphi)$ with $|\theta - \theta_p| < r$ and $|\varphi - \varphi_q| < r$, the Jacobian remains invertible and
\[
|\det\bm{J}_{\mathbf{g}}(\theta, \varphi)| \geq \frac{1}{2} |\det\bm{J}_0| > 0.
\]
By the inverse function theorem, this implies that $\mathbf{g}$ is a $C^1$ diffeomorphism from the rectangle
\[
\mathcal{R} := [\theta_p - \Delta\theta,\,\theta_p + \Delta\theta] \times [\varphi_q - \Delta\varphi,\,\varphi_q + \Delta\varphi]
\]
onto its image $\mathbf{g}(\mathcal{R})$, for any $\Delta\theta, \Delta\varphi < r$.

Now, consider the image of $\mathcal{R}$ under the first-order (linear) Taylor expansion of $\mathbf{g}$ at $(\theta_p, \varphi_q)$:
\[
\mathbf{g}(\theta, \varphi) \approx \mathbf{g}(\theta_p, \varphi_q) 
+ \bm{J}_0
\begin{bmatrix}
    \theta - \theta_p \\
    \varphi - \varphi_q
\end{bmatrix}.
\]
The rectangle $\mathcal{R}$ in $(\theta, \varphi)$-space, defined by $|\theta - \theta_p| \leq \Delta\theta$, $|\varphi - \varphi_q| \leq \Delta\varphi$, is mapped by this linear transformation to a parallelogram in the tangent plane. Specifically, the sides of the parallelogram are given by the vectors
\[
\bm{J}_0 \begin{bmatrix} \Delta\theta \\ 0 \end{bmatrix}
\quad\text{and}\quad
\bm{J}_0 \begin{bmatrix} 0 \\ \Delta\varphi \end{bmatrix},
\]
which are the images of the edges of the rectangle under the Jacobian. Because $\bm{J}_0$ is invertible, this parallelogram is nondegenerate: it has positive area and does not collapse to a line or a point.

Therefore, for sufficiently small bins, $\mathbf{g}$ is a $C^1$ diffeomorphism from $\mathcal{R}$ onto its image, and the linear image of $\mathcal{R}$ is a nondegenerate parallelogram in the tangent plane as claimed.

\qed

\subsection{Proof of Theorem~\ref{thm:GWA-error-full}}
\label{app:thm1} 
Set
\[
\bm{J} = \bm{J}_{\mathbf{g}}(\theta_p, \varphi_q), \qquad
\boldsymbol{\delta} = \begin{bmatrix} \theta' - \theta_p \\ \varphi' - \varphi_q \end{bmatrix}.
\]

We split the error:
\[
w^{\mathrm{exact}} - \widetilde{w} = E_{\mathrm{geo}} + E_{\mathrm{win}},
\]
with
\begin{align*}
E_{\mathrm{geo}} &= \iint_{\mathcal{R}} [f(\mathbf{g}(\theta',\varphi'))\,|\det\bm{J}_{\mathbf{g}}(\theta',\varphi')|\\
&\qquad\qquad- f(\mathbf{g}(\theta_p, \varphi_q) + \bm{J}\boldsymbol{\delta})\,|\det\bm{J}| ]\,\mathrm{d}\theta'\mathrm{d}\varphi', \\
E_{\mathrm{win}} &= \iint_{\mathcal{R}} f(\mathbf{g}(\theta_p, \varphi_q) + \bm{J}\boldsymbol{\delta})\,|\det\bm{J}|\,\mathrm{d}\theta'\mathrm{d}\varphi' \\
&\qquad\qquad- \int_{\mathbb{R}^2} f(\mathbf{y})\,W(\mathbf{y})\,\mathrm{d}\mathbf{y}.
\end{align*}

We first bound $E_{\mathrm{geo}}$. Add and subtract  $f(\mathbf{g}(\theta_p, \varphi_q) + \bm{J}\boldsymbol{\delta})\,|\det\bm{J}_{\mathbf{g}}(\theta',\varphi')|$ inside the integrand, we have
\begin{align*}
\small
&E_{\mathrm{geo}} \\
&= \iint_{\mathcal{R}} \Big( [f(\mathbf{g}(\theta',\varphi')) - f(\mathbf{g}(\theta_p, \varphi_q) + \bm{J}\boldsymbol{\delta}) ]\,|\det\bm{J}_{\mathbf{g}}(\theta',\varphi')| \\
&\qquad+ f(\mathbf{g}(\theta_p, \varphi_q) + \bm{J}\boldsymbol{\delta})\, [|\det\bm{J}_{\mathbf{g}}(\theta',\varphi')| 
- |\det\bm{J}| ] \Big)\,\mathrm{d}\theta'\mathrm{d}\varphi'.
\end{align*}
By the triangle inequality,
$|E_{\mathrm{geo}}| \leq I_1 + I_2$,
where
{\small
\begin{align*}
I_1 &= \iint_{\mathcal{R}} |f(\mathbf{g}(\theta',\varphi')) - f(\mathbf{g}(\theta_p, \varphi_q) + \bm{J}\boldsymbol{\delta})|\,|\det\bm{J}_{\mathbf{g}}(\theta',\varphi')|\,\mathrm{d}\theta'\mathrm{d}\varphi', \\
I_2 &= \iint_{\mathcal{R}} |f(\mathbf{g}(\theta_p, \varphi_q) + \bm{J}\boldsymbol{\delta})|\,| |\det\bm{J}_{\mathbf{g}}(\theta',\varphi')| - |\det\bm{J}| |\,\mathrm{d}\theta'\mathrm{d}\varphi'.
\end{align*}}
By Lemma~\ref{lem:g-regularity-inline}, we have 
\[
|f(\mathbf{g}(\theta',\varphi')) - f(\mathbf{g}(\theta_p, \varphi_q) + \bm{J}\boldsymbol{\delta})| 
\leq L_f C_g \|\boldsymbol{\delta}\|^2,
\]
for some constant $L_f$ related to the Lipschitz constant of $f(\cdot)$. Also, by Lemma~\ref{lem:local-invertibility-inline}, $|\det\bm{J}_{\mathbf{g}}(\theta',\varphi')|$ is bounded above by some constant $C_J$ in $\mathcal{R}$. Therefore,
\[
I_1 \leq L_f C_g C_J \iint_{\mathcal{R}} \|\boldsymbol{\delta}\|^2\,\mathrm{d}\theta'\mathrm{d}\varphi'.
\]

On the compact image of the bin, $f$ is bounded, i.e.,
$|f(\mathbf{g}(\theta_p, \varphi_q) + \bm{J}\boldsymbol{\delta})| \leq C_f$
for some $C_f$. By the smoothness of $\mathbf{g}(\cdot)$, $|\det\bm{J}_{\mathbf{g}}|$ is $C^1$, so by Taylor expansion and the
symmetry of the bin $\mathcal{R}$ about $(\theta_p, \varphi_q)$, we have
\[
| |\det\bm{J}_{\mathbf{g}}(\theta', \varphi')| - |\det\bm{J}| | \leq C_{J,2} \|\boldsymbol{\delta}\|^2
\]
for some $C_{J,2}$ and all $(\theta',\varphi') \in \mathcal{R}$. Thus,
$I_2 \leq C_f C_{J,2} \iint_{\mathcal{R}} \|\boldsymbol{\delta}\|^2\,\mathrm{d}\theta'\mathrm{d}\varphi'$.
Since $\iint_{\mathcal{R}} \|\boldsymbol{\delta}\|^2\,\mathrm{d}\theta'\mathrm{d}\varphi' = O(\Delta\theta^3 + \Delta\varphi^3)$,
adding the bounds for $I_1$ and $I_2$ yields
$|E_{\mathrm{geo}}| \leq C_3 (\Delta\theta^3 + \Delta\varphi^3)$
for some constant $C_3>0$.

Next we bound $E_{\mathrm{win}}$.  Let $\mathbf{g}_0 = \mathbf{g}(\theta_p, \varphi_q)$. Substitute $\bm{y} = \mathbf{g}_0 + \bm{J}\boldsymbol{\delta}$, so
\[
E_{\mathrm{win}} = \int_{\mathcal{P}} f(\bm{y})\,\mathrm{d}\bm{y} - \int_{\mathbb{R}^2} f(\bm{y})\,W(\bm{y})\,\mathrm{d}\bm{y},
\]
where $\mathcal{P} = \bm{J}(\mathcal{R}_0) + \mathbf{g}_0$ is the parallelogram image of the bin. Expand $f$ in a Taylor series at $\mathbf{g}_0$:
\begin{align*}
E_{\mathrm{win}} 
&= \sum_{k=0}^2 \left( \int_{\mathcal{P}} T_k(\bm{y})\,\mathrm{d}\bm{y} - \int_{\mathbb{R}^2} T_k(\bm{y})\,W(\bm{y})\,\mathrm{d}\bm{y} \right) \\
&\quad + \left[ \int_{\mathcal{P}} R_3(\bm{y})\,\mathrm{d}\bm{y} - \int_{\mathbb{R}^2} R_3(\bm{y})\,W(\bm{y})\,\mathrm{d}\bm{y} \right],
\end{align*}
where $T_k$ is the $k$-th degree Taylor term and $|R_3(\bm{y})| \leq L_3 \|\bm{y} - \mathbf{g}_0\|^3$ for some $L_3 > 0$.

By the GWA construction, $W$ is a Gaussian window with area and covariance \emph{exactly} matching those of the parallelogram $\mathcal{P}$ while the first-order terms are both zero due to symmetry. As a result,
\[
|E_{\mathrm{win}}| \leq \int_{\mathcal{P}} |R_3(\bm{y})|\,\mathrm{d}\bm{y} + \int_{\mathbb{R}^2} |R_3(\bm{y})|\,W(\bm{y})\,\mathrm{d}\bm{y}.
\]
Since $\|\bm{y} - \mathbf{g}_0\| \leq h = O(\Delta\theta + \Delta\varphi)$ and the area of $\mathcal{P}$ is $O(\Delta\theta\Delta\varphi)$, both integrals are $O(\Delta\theta^4 + \Delta\varphi^4 + \Delta\theta^2\Delta\varphi^2)$.
Therefore, for some $C_4 > 0$,
\[
|E_{\mathrm{win}}| \leq C_4 (\Delta\theta^4 + \Delta\varphi^4 + \Delta\theta^2\Delta\varphi^2).
\]

Adding the bounds for $E_{\mathrm{geo}}$ and $E_{\mathrm{win}}$ completes the proof.
\qed

\bibliographystyle{IEEEtran}
\bibliography{refs}

\end{document}